\documentclass{acmsmall} 

\usepackage{amsmath,amsfonts,amssymb}

\usepackage{bm}

\usepackage{xspace}

\usepackage{paralist}

\usepackage{graphicx}
\usepackage{subfig}

\usepackage{boxedminipage}

\usepackage[colorlinks,urlcolor=blue,citecolor=blue,linkcolor=blue]{hyperref}

\usepackage{algorithm}
\usepackage{algpseudocode}

\def\submit{0}   
\ifnum\submit=1
	\newcommand{\full}[1]{}
	\newcommand{\confer}[1]{#1}
\else
	\newcommand{\full}[1]{#1}
	\newcommand{\confer}[1]{}
\fi

\newtheorem{claim}[theorem]{Claim}
\newtheorem{empirical}[theorem]{Empirical Observation}

\newenvironment{myproof}{\par\textsc{Proof.}}{}
\newcommand{\myproofend}{\hfill \qed}

\newcommand{\Real}{\mathbb{R}}

\newcommand{\bT}{{\bf T}}

\newcommand{\cE}{{\cal E}}

\newcommand{\nint}[1]{\lfloor #1\rceil}
\newcommand{\flo}[1]{\lfloor #1\rfloor}

\newcommand{\pr}{{\rm Pr}}

\newcommand{\EX}{\hbox{\bf E}}

\newcommand{\odeg}[1]{\text{\rm deg}(#1)}

\newcommand{\halpha}{\widehat{\alpha}}
\newcommand{\htau}{\widehat{\tau}}
\newcommand{\nuone}{\nu_1}
\newcommand{\nutwo}{\nu_2}
\newcommand{\nuthr}{\nu_3}
\newcommand{\nufour}{\nu_4}

\newcommand{\Sec}[1]{\hyperref[sec:#1]{\S\ref*{sec:#1}}} %
\newcommand{\App}[1]{\hyperref[sec:#1]{Appendix~\ref*{sec:#1}}} %
\newcommand{\Eqn}[1]{\hyperref[eq:#1]{(\ref*{eq:#1})}} %
\newcommand{\Fig}[1]{\hyperref[fig:#1]{Figure\,\ref*{fig:#1}}} %
\newcommand{\Tab}[1]{\hyperref[tab:#1]{Table\,\ref*{tab:#1}}} %
\newcommand{\Thm}[1]{\hyperref[thm:#1]{Theorem\,\ref*{thm:#1}}} %
\newcommand{\Lem}[1]{\hyperref[lem:#1]{Lemma\,\ref*{lem:#1}}} %
\newcommand{\Prop}[1]{\hyperref[prop:#1]{Prop.~\ref*{prop:#1}}} %
\newcommand{\Cor}[1]{\hyperref[cor:#1]{Cor.~\ref*{cor:#1}}} %
\newcommand{\Def}[1]{\hyperref[def:#1]{Defn.~\ref*{def:#1}}} %
\newcommand{\Alg}[1]{\hyperref[alg:#1]{Alg.~\ref*{alg:#1}}} %
\newcommand{\Ex}[1]{\hyperref[ex:#1]{Ex.~\ref*{ex:#1}}} %
\newcommand{\Clm}[1]{\hyperref[clm:#1]{Claim~\ref*{clm:#1}}} %

\newcommand{\M}[1]{{{{\MakeUppercase{#1}}}}} %
\newcommand{\qtext}[1]{\quad\text{#1}\quad} %

\begin{document}

\markboth{C. Seshadhri, A. Pinar, T. G. Kolda}{Improving Stochastic Kronecker Graphs}

\title{An In-Depth Analysis of Stochastic Kronecker Graphs}

\author{C. SESHADHRI
\affil{Sandia National Laboratories} 
ALI PINAR
\affil{Sandia National Laboratories} 
TAMARA G. KOLDA
\affil{Sandia National Laboratories}}

\begin{abstract}
  Graph analysis is playing an increasingly important role in science
  and industry. Due to numerous limitations in sharing real-world
  graphs, models for generating massive graphs are critical for
  developing better algorithms. In this paper, we analyze the
  stochastic Kronecker graph model (SKG), which is the
  foundation of the Graph500 supercomputer benchmark due to its
  favorable properties and easy parallelization. Our goal is to
  provide a deeper understanding of the parameters and properties of this model so
  that its functionality as a benchmark is increased.  We develop a
  rigorous mathematical analysis that shows this model \emph{cannot}
  generate a power-law distribution or even a lognormal distribution.
  However, we formalize an enhanced version of the SKG
  model that uses random noise for smoothing.
  We prove both in theory and in practice that this enhancement leads to a lognormal
  distribution. Additionally, we provide a precise analysis of 
  isolated vertices, showing that the graphs that are produced by SKG might be quite different than
  intended. For example, between 50\% and 75\% of the vertices in the
  Graph500 benchmarks will be isolated.  Finally, we show that this
  model tends to produce extremely small core numbers 
  (compared to most social networks and other real graphs) for common parameter choices. 

\end{abstract}



\begin{bottomstuff}
This work was funded by the applied mathematics program at the United
States Department of Energy and performed at Sandia National
Laboratories, a multiprogram laboratory operated by Sandia
Corporation, a wholly owned subsidiary of Lockheed Martin Corporation,
for the United States Department of Energy's National Nuclear Security
Administration under contract DE-AC04-94AL85000.

Author's addresses:
C. Seshadhri, A. Pinar, T. Kolda, Sandia National Laboratories, Livermore, CA 94551.
\end{bottomstuff}

\maketitle

\section{Introduction}
\label{sec:intro}

The role of graph analysis is  becoming increasingly important in science
and industry because of the prevalence of graphs in diverse scenarios
such as social networks, the Web, power grid networks, and even
scientific collaboration studies. Massive
graphs occur in a variety of situations, and we need to design better and faster 
algorithms in order to study them. However, it can be difficult
to access to 
informative large graphs in order to test our algorithms. 
Companies like Netflix, AOL, and Facebook have vast arrays
of data but cannot share it due to legal or copyright
issues\footnote{For example, Netflix opted not to pursue the Netflix Prize sequel due to concerns about lawsuits; see \url{http://blog.netflix.com/2010/03/this-is-neil-hunt-chief-product-officer.html}}.
Moreover, graphs with billions of vertices cannot be
communicated easily due to their sheer size.

As was noted in \cite{ChFa06}, good \emph{graph models} are extremely important for
the study and algorithmics of real networks. Such a model should be fairly
easy to implement  and have few parameters,  while exhibiting the common properties  of  real networks. 
Furthermore,  models are needed  to test algorithms and architectures designed for large graphs. 
But the theoretical and research benefits are also
obvious:  gaining  insight into the properties and processes that
create real networks.

The \emph{stochastic Kronecker graph} (SKG) \cite{LeFa07,LeChKlFa10}, a generalization
of the \emph{recursive matrix} (R-MAT) model \cite{ChZhFa04}, has been proposed
for these purposes. It has very few parameters and can generate large graphs
quickly. Indeed, it is one of the few models that can generate
graphs  fully in \emph{parallel}. It has been empirically observed to have
interesting real-network-like properties. We stress that this is not
just of theoretical or academic interest---this model has been
chosen to create graphs for the Graph500 supercomputer
benchmark \cite{Graph500}. 

It is important
to know how the parameters of this model affect various properties
of the graphs. 
We stress that a mathematical analysis is important for understanding
the inner working of a model. We quote Mitzenmacher~\cite{Mi06}:
``I would argue, however, that without validating a model it is
not clear that one understands the underlying behavior and therefore how the
behavior might change over time. It is not enough to plot data and demonstrate
a power law, allowing one to say things about current behavior; one wants to
ensure that one can accurately predict future behavior appropriately, and that
requires understanding the correct underlying model."

\subsection{Notation and Background}
\label{sec:rmat}

We explain the SKG model and notation.
Our goal is to generate a directed graph $G = (V, E)$ with $n = |V|$ nodes and $m = |E|$ edges. 
The general form of the SKG model allows for an arbitrary square generator matrix and 
assumes that $n$ is a power of its size. 
Here, we focus on the $2 \times 2$ case (which is equivalent to R-MAT), defining the generating matrix as
\begin{displaymath}
  \M{T} =
  \begin{bmatrix} 
    t_1 & t_2 \\ 
    t_3 & t_4 
  \end{bmatrix}
  \qtext{with}
  t_1 + t_2 + t_3 + t_4 = 1 \ \textrm{and} \ \min_i t_i > 0.
\end{displaymath}
We assume that $n=2^\ell$ for some integer $\ell>0$. 
\emph{For the sake of cleaner formulae, we assume
that $\ell$ is even in our analyses.}
Each edge is inserted according to the probabilities
defined by
\begin{displaymath}
  \M{P} = 
  \underbrace{
    \M{T} \otimes \M{T} \otimes \cdots \otimes \M{T}
  }_{\text{$\ell$ times}},
\end{displaymath}
where $\otimes$ denotes the  Kronecker product operation.
In practice, the matrix $\M{P}$ is never formed explicitly. 
Instead, each edge is inserted as follows.
Divide the adjacency matrix into
four quadrants, and choose one of them with the corresponding probability $t_1, t_2, t_3$, or $t_4$.
Once a quadrant is chosen, repeat this recursively in that quadrant. Each time we iterate,
we end up in a square submatrix whose dimensions are exactly halved. After $\ell$ iterations,
we reach a single cell of the adjacency matrix, and an edge is inserted. 
It should be noted that here we  take a slight liberty in requiring the entries of $T$ to sum to 1. In fact, the SKG model as defined in \cite{LeChKlFa10} works with the matrix $mP$, which is considered the matrix of probabilities for the existence of each individual edge (though it might be more accurate to think of it as an expected value).

Note that all edges can be inserted in parallel. This is one of the major advantages
of the SKG model and why it is appropriate for generating large supercomputer benchmarks.

For convenience, we also define some derivative parameters that will
be useful in subsequent discussions. We let $\Delta = m/n$ denote the
\emph{average degree} and let $\sigma = t_1+t_2-0.5$ denote the
\emph{skew}. 
The parameters of the SKG model are summarized in \Tab{params1}.

\begin{table}[t]
  \caption{Parameters for SKG models}
  \label{tab:params1}
  \centering
  \begin{boxedminipage}{.9\textwidth}
    \textsc{Primary Parameters}
    \begin{itemize}
    \item $T = \begin{bmatrix} t_1 & t_2 \\ t_3 & t_4 \end{bmatrix} = $ 
      generating matrix with $t_1 + t_2 + t_3 + t_4 = 1$     
    \item $\ell = $ number of levels (assumed even for analysis)
    \item $m = $ number of edges
    \end{itemize}
    \textsc{Derivative Parameters} 
    \begin{itemize}
    \item $n = 2^{\ell} = $ number of nodes 
    \item  $\Delta = m/n = $ average degree
    \item  $\sigma = t_1 + t_2 - 0.5 = $ skew %
    \end{itemize}
  \end{boxedminipage}
\end{table}

\subsection{Our Contributions}
\label{sec:contributions}

Our overall contribution is to provide a thorough study of the properties
of SKG and show how the parameters affect these properties. 
We focus on the degree distribution,
the number of (non-isolated nodes), the core sizes, and the trade-offs
in these various goals. We give rigorous mathematical theorems and proofs 
explaining the degree distribution of SKG, a noisy version of SKG, and the number
of isolated vertices.

\begin{asparaenum}
\item {\bf Degree distribution:} We provide a rigorous mathematical analysis of the degree distribution of SKG. The degree distribution has often been claimed to be power-law, 
or sometimes  lognormal \cite{ChZhFa04,LeChKlFa10,KiLe10}. 
Kim and Leskovec \cite{KiLe10} prove that the degree distribution has some lognormal characteristics.
Gro\"er et al.\@ \cite{GrSuPo10} give exact series expansions for the degree distribution, and express it
as a mixture of normal distributions.
This provides a qualitative explanation for the oscillatory behavior 
of the degree distribution (refer to \Fig{noisy_degdist_graph500}). Since the distribution is quite far from being truly lognormal, there has been no simple closed form expression that closely approximates it.
We fill this gap by providing a complete mathematical description. 
We prove that SKG \emph{cannot} generate a power law distribution, or
even a lognormal distribution. It is most accurately characterized as fluctuating between a lognormal distribution 
and an exponential tail. We provide a simple formula that approximates
the degree distribution.

\item {\bf Noisy SKG:} It has been mentioned in passing \cite{ChZhFa04}
that adding noise to SKG at each level 
smoothens the degree distribution, but this
has never been formalized or studied. 
We define a specific noisy version of SKG (NSKG).
We prove theoretically and empirically that NSKG
leads to a lognormal distribution. (We give some experimental
results showing a naive addition of noise does \emph{not} work.)
The lognormal distribution is important since it has been observed in real data \cite{BiFaKo01,PeFlLa+02,Mi03,ClShNe09}.
One of the major benefits of our enhancement 
is that only $\ell$ additional random numbers are needed in total.
Using Graph500 parameters, \Fig{noisy_degdist_graph500} 
plots the degree distribution of a (standard) SKG and NSKG for two levels of (maximum) noise.
We can clearly see that noise dampens the oscillations, leading to
a lognormal distribution. We note that though the modification of NSKG is straightforward, 
the reason why it works is not. It involves an intricate
mathematical analysis, which may be of theoretical interest in itself.

\begin{figure}
  \centering
  \includegraphics[width=.45\textwidth,trim=0 0 0 5,clip]{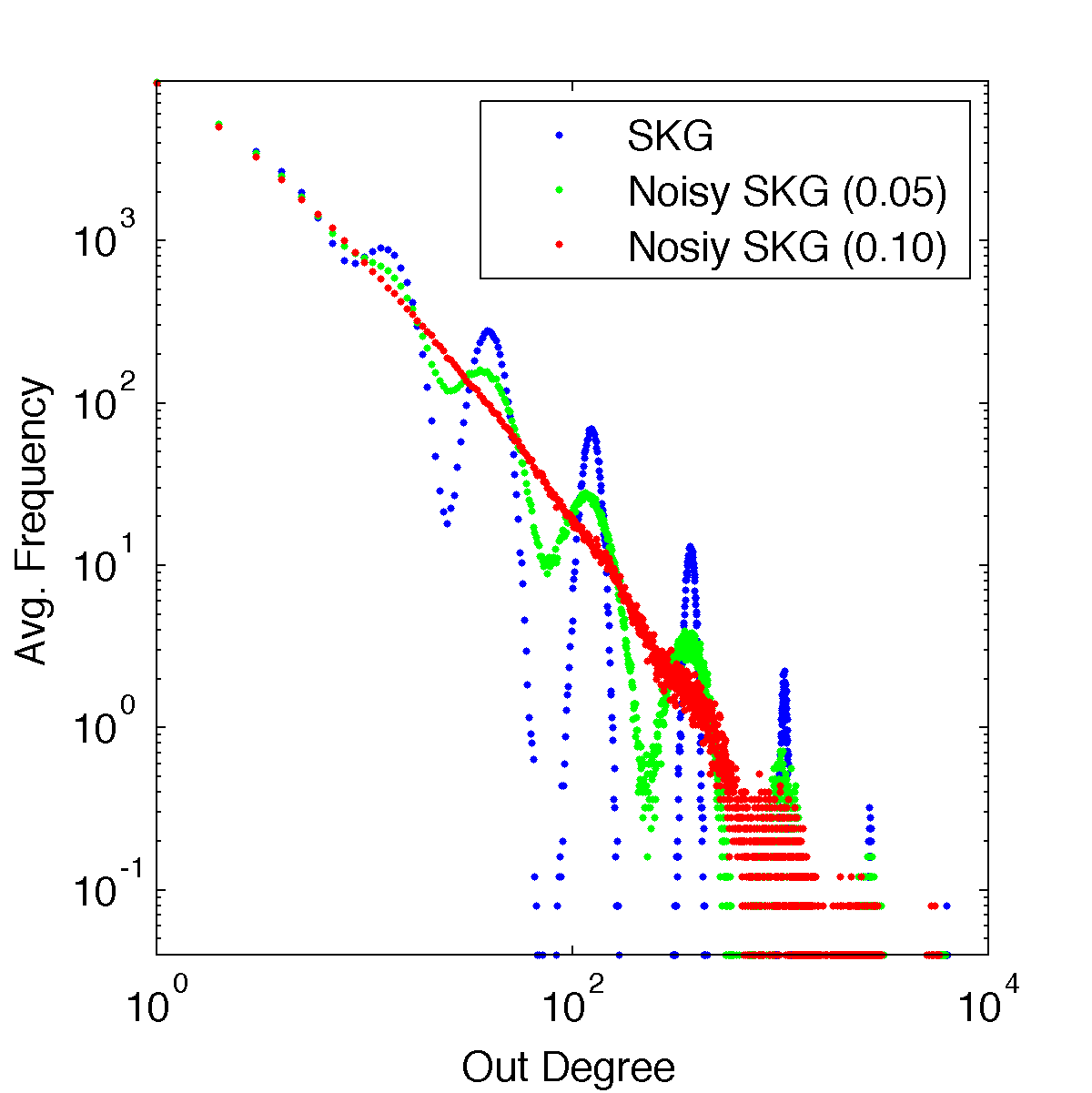}
  \caption{Comparison of degree distributions (averaged over 25
    instances) for SKG and two noisy variations, using the $T$ from the Graph500 Benchmark parameters 
    with $\ell=16$.}
  \label{fig:noisy_degdist_graph500}
\end{figure}

\item {\bf Isolated vertices:} An isolated vertex is one that has no
edges incident to it (and hence is not really part of the output graph).
We provide a formula that accurately estimates the fraction of isolated vertices. 
We discover the surprising result that in the Graph500 benchmark graphs,
50-75\% vertices are isolated; see \Tab{isol-500}. This is a major concern for the
benchmark, since the massive graph generated has a much reduced size. Furthermore,
the average degree is now much higher than expected.

\begin{table}[t] 
  \caption{Expected percentage of isolated vertices and repeat edges, 
    along with average degree of \emph{non-isolated} nodes for the Graph500 benchmark. 
    Excluding the isolated vertices results in a much higher average degree than 
    the value of 16 that is specified by the benchmark.} 
\centering %
\begin{tabular}{cccc} %
\hline 
$\ell$ & \% Isolated Nodes & \% Repeat Edges & Avg.\@ Degree \\  \hline
26 & 51 & 1.2 & 32 \\
29 & 57 & 0.7 & 37 \\
32 & 62 & 0.4 & 41 \\
36 & 67 & 0.2 & 49 \\
39 & 71 & 0.1 & 55 \\
42 & 74 & 0.1 & 62 \\
\hline %
\end{tabular}
\label{tab:isol-500}
\end{table}

\item {\bf Core numbers:} The study of $k$-cores is an important tool used to study the structure of social
networks because it is a mark of the connectivity and special processes that generate these graphs \cite{ChFa06,KuNoTo10,AlDaBa+08,GkMi03,GoDo06,CaHaKi+07,AnCh09}. We 
empirically show how the core numbers have unexpected correlations with SKG parameters. 
We observed that for most of the current SKG parameters used for modeling real graphs, max core numbers
are extremely small (much smaller than most corresponding real graphs). We show how modifying the matrix $T$
affects core numbers.
Most strikingly, we observe that changing $T$ to increase the max core number actually leads to an increase 
in the fraction of isolated vertices.
\end{asparaenum}

\subsection{Influence on Graph500 benchmark} \label{sec:graph500}

Our results have been communicated to the Graph500 steering committee, who have found
them useful in understanding the Graph500 benchmark. The oscillations in the degree distribution
of SKG was a major concern for the committee. 
Our proposed NSKG
model has been implemented in the current Graph500 code\footnote{The file 
generator/graph-generator.c in the most recent version as of July 2012 (2.1.4) has the implementation,
with a variable SPK\_NOISE\_LEVEL controlling the NSKG noise. Available at \url{http://www.graph500.org/sites/default/files/files/graph500-2.1.4.tar.bz2}}.

Our analysis also solves the mystery of isolated vertices and how they are related to the SKG parameters. Members
of the steering committee had observed that the number of isolated vertices varied greatly
with the matrix $T$, but did not have an explanation for this. 

\subsection{Parameters for empirical study}

Throughout the paper, we discuss a few sets of SKG parameters.
The first is the Graph500 benchmark \cite{Graph500}. The other two are 
parameters used in \cite{LeChKlFa10} to model a
co-authorship network (CAHepPh) and a web graph (WEBNotreDame).
We list these parameters here for later reference.

\begin{asparaitem}
	\item \textbf{Graph500}: $T=[0.57, 0.19; 0.19, 0.05]$, $\ell \in \{$26, 29, 32, 36, 39, 42$\}$, and $m = 16 \cdot 2^\ell$.
	\item \textbf{CAHepPh}: $T = [0.42, 0.19; 0.19, 0.20]$, $\ell = 14$, and $m = 237,010$.
	\item \textbf{WEBNotreDame}\footnote{In \cite{LeChKlFa10},
	$\ell$ was 19. We make it even because, for the sake of presentation, we perform experiments
	and derive formulae for even $\ell$.}: $T = [0.48, 0.20; 0.21, 0.11]$, $\ell = 18$, and $m = 1,497,134$.
\end{asparaitem}

\section{Previous Work}
\label{sec:related}

The R-MAT model was
defined by Chakrabarti et al.\@ \cite{ChZhFa04}. The general and more powerful
SKG model was introduced by Leskovec et al.\@ \cite{LeChKlFa05} and fitting
algorithms were proposed by Leskovec and Faloutsos \cite{LeFa07} (combined in \cite{LeChKlFa10}).
This model has generated significant interest and notably was
chosen for the Graph500 benchmark \cite{Graph500}. 
Kim and Leskovec~\cite{KiLe10} defined the Multiplicative Attribute Graph (MAG) model, 
a generalization of SKG where each level may have a different matrix $T$.
They suggest that certain configurations of these matrices could lead
to power-law distributions. 

Since the appearance of the SKG model, there have been analyses of its properties.
The original paper \cite{LeChKlFa10} provides some basic theorems and empirically show
a variety of properties. Mahdian and Xu~\cite{MaXu10} specifically study how
the model \emph{parameters} affect the graph properties. They show phase transition
behavior (asymptotically) for occurrence of a large connected component and shrinking diameter.
They also initiate a study of isolated vertices. 
When the SKG parameters satisfy a certain condition,
the number of isolated vertices approaches $n$; however, their
theorems do not help predict 
the number of isolated vertices for a given setting of SKG.
In the analysis of the MAG model \cite{KiLe10}, it is shown that the SKG degree
distribution has some lognormal characteristics. (Lognormal
distributions have been observed in real data \cite{BiFaKo01,PeFlLa+02,ClShNe09}.
Mitzenmacher~\cite{Mi03} gives a survey of lognormal distributions.)

Sala et al.\@ \cite{SaCaWiZa10} perform an extensive
empirical study of properties of graph models, including SKG.
Miller et al.\@ \cite{MiBlWo10} show that they can detect anomalies embedded in an SKG.
Moreno et al.\@ \cite{MoKiNeVi10} study the distributional properties of families of SKG.

As noted in \cite{ChZhFa04}, the SKG generation procedure may give repeated edges. Hence,
the number of edges in the graph differs slightly from the number
of  insertions (though, in practice, this is barely 1\% for Graph500).
Gro\"er et al.\@ \cite{GrSuPo10} 
prove that the number of vertices of a given degree is asymptotically normally
distributed, and  provide algorithms to compute the expected number of edges
in the graph (as a function of the number of insertions) and the expected
degree distribution.

\section{Degree Distribution}
\label{sec:degdist}

In this section, we analyze the degree distribution of SKG, which are known to 
follow a multinomial distribution.
While an  exact expression for this distribution can be written, this is unfortunately a 
complicated sum of binomial coefficients. Studying the log-log plots
of the degree distribution, one sees a general heavy-tail like behavior,
but there are large oscillations. The degree distribution is not monotonically decreasing. Refer to \Fig{degdist} for some examples of SKG degree distributions (plotted in log-log scale).
Gro\"er et al.\@ \cite{GrSuPo10} show that the degree distribution behaves
like the sum of Gaussians, giving some intuition for the oscillations.
Recent work of Kim and Leskovec~\cite{KiLe10} provide some mathematical analysis
explaining connections to a lognormal distribution. 
But many questions remain. What does the distribution
oscillate between? Is the distribution bounded below by a power law?
Can we approximate the distribution with a simple closed form
function? None of these questions have satisfactory answers.

Our analysis gives a precise explanation
for the SKG degree distribution. We prove that the SKG degree
distribution oscillates between a lognormal and exponential tail.
We provide plots and experimental results to support
more intuition for our theorems.  

The oscillations are a disappointing feature of SKG. Real degree distributions
do not have large oscillations (to the contrary, they are monotonically decreasing), and more importantly, do not have any
exponential tail behavior. This is a major issue both for modeling
and benchmarking purposes since degree distribution is one of the primary 
characteristics that distinguishes real networks.

In order to rectify the oscillations, we apply a certain model of noise and provide both mathematical and empirical evidence that this
``straightens out" the degree distribution. This is discussed
in \Sec{nskg}. Indeed, small amounts of
noise lead to a degree distribution that is predominantly lognormal.
This also shows an appealing aspect of our degree distribution
analysis. We can naturally explain how noise affects the degree
distribution and give explicit bounds on these affects.

We make a caveat here. Technically, the SKG model creates \emph{multigraphs}, since
there can be repeated edges. Our theorems and expressions will deal with degree
distributions of this multigraph. Conventionally, this is reduced to a simple graph
by removing repeated edges. Gro\"er et al.\@ \cite{GrSuPo10} give details expressions
and explanations relating the degree distributions on the multigraph and the induced simple
graph. Our empirical results show that for a variety of parameters (including the Graph 500 setting),
our theorems match the degree distribution of the underlying simple graph. Simple
graphs are used in all empirical studies.

\subsection{Notation}

The $\ell$-bit binary representation of the vertices, numbered 0 to $n-1$, 
provides a straightforward way to partition the vertices. Specifically,
each vertex has a binary representation and therefore corresponds to an element of the
boolean hypercube $\{0,1\}^\ell$. We can partition the vertices into 
\emph{slices}, where each slice consists of vertices whose 
representations have the same number of zeros\footnote{There are usually referred to as the \emph{levels of the boolean hypercube}.
In the SKG literature, levels is used to refer to $\ell$, and hence we use
a different term.}. Recall that we assume $\ell$ is even.
For $r \in [-\ell/2,\ell/2]$, we say that \emph{slice $r$}, denoted $S_r$, consists of all vertices whose binary 
representations have exactly $(\ell/2+r)$ zeros. 

These binary representations and slices are intimately connected with  edge insertions in the SKG model.
For each insertion, we are trying to randomly choose a source-sink pair. First, let
us simply choose the first bit (of the representations) of the source
and the sink. Note that there are 4 possibilities (first bit for source, second for sink): 00, 01, 10,  and 11. 
We choose one of the combinations with probabilities $t_1, t_2, t_3$, and $t_4$ respectively. This
fixes the first bit of the source and sink. We perform this procedure again
to choose the second bit of the source and sink. Repeating $\ell$ times,
we finally decide the source and sink of the edge.
Note that as $|r|$ becomes smaller, a vertex in an $r$-slice tends to have a higher degree. 

For a real number $x$, we use $\nint{x}$ to denote the closest integer to $x$.
There are certain quantities that will be important in our analysis. These are summarized in \Tab{params2}. 

\begin{table}[tbp]
  \caption{Parameters for Analysis of SKG models}
  \label{tab:params2}
  \centering
  \begin{boxedminipage}{.9\textwidth}
    \textsc{General Quantities}
    \begin{itemize}
    \item $\tau = (1+2\sigma)/(1-2\sigma)$
    \item $\lambda = \Delta(1-4\sigma^2)^{\ell/2}$
    \item $r \in \{ -\ell/2,\dots,\ell/2 \} $ denotes a slice index
    \item $d $ denotes a degree (typically assumed $< \sqrt{n}$)
    \item $\odeg{v}$ = outdegree of node $v$
    \item $S_r = $ set of nodes whose binary representation have exactly $\ell/2 + r$ zeros
    \end{itemize}
    \textsc{Quantities Associated with Degree $d$}
    \begin{itemize}
    \item $X_d = $ random variable for the number of vertices of outdegree $d$
    \item $\theta_d = \ln (d/\lambda)/\ln \tau$
    \item $\Gamma_d = \nint{\theta_d}$ (nearest integer to $\theta_d$)
    \item $\gamma_d = |\theta_d - \Gamma_d| \in [0,0.5]$
    \item $r_d = \lfloor \theta_d \rfloor$ (only interesting for $r_d < \ell/2$)
    \item $\delta_d = \theta_d - r_d$
    \end{itemize}
  \end{boxedminipage}
\end{table}

Our results are fundamentally asymptotic in nature, so we explain the assumptions on $T$
and the implicit assumptions of our results. We assume $T$ to be a \emph{fixed} matrix
with the following conditions. All entries are positive and strictly less than $1$.
The number $t_1$ is the largest entry, and $\min(t_1 + t_2, t_1 + t_3) > 1/2$. 
This ensures that $\sigma \in (0,1/2)$, $\tau$ is positive
and finite, and $\lambda$ is non-zero. We want to note that these conditions are satisfied by all SKG parameters that
have been used to generate realistic graph instances, to the best of our knowledge. 
Indeed, when $\sigma = 1/2$, the degree distribution is Poisson.

We fix the matrix $T$ and average degree $\Delta > 1$, and think of $\ell$ as increasing. The asymptotics hold for an increasing $\ell$.
Note that since $n = 2^\ell$, this means that $n$ and $m$ are also increasing. We use $o(1)$ as a shorthand for a quantity that is negligible
as $\ell \rightarrow \infty$. Typically,
this converges to zero rapidly as $\ell$ increases.
Given two quantities or expressions $A$ and $B$, $A = (1 \pm o(1))B$ will
be shorthand for $A \in [(1-o(1))B, (1+o(1))B]$.

As we mentioned earlier, all our results are for the SKG multigraph. For convenience,
we will just refer to this a graph.

\subsection{Explicit formula for degree distribution} \label{sec:deg-an}

\begin{figure*}[ptb]
  \centering
  \subfloat[CAHepPh]{
  \includegraphics[width=.45\textwidth,trim=0 0 0 5,clip]{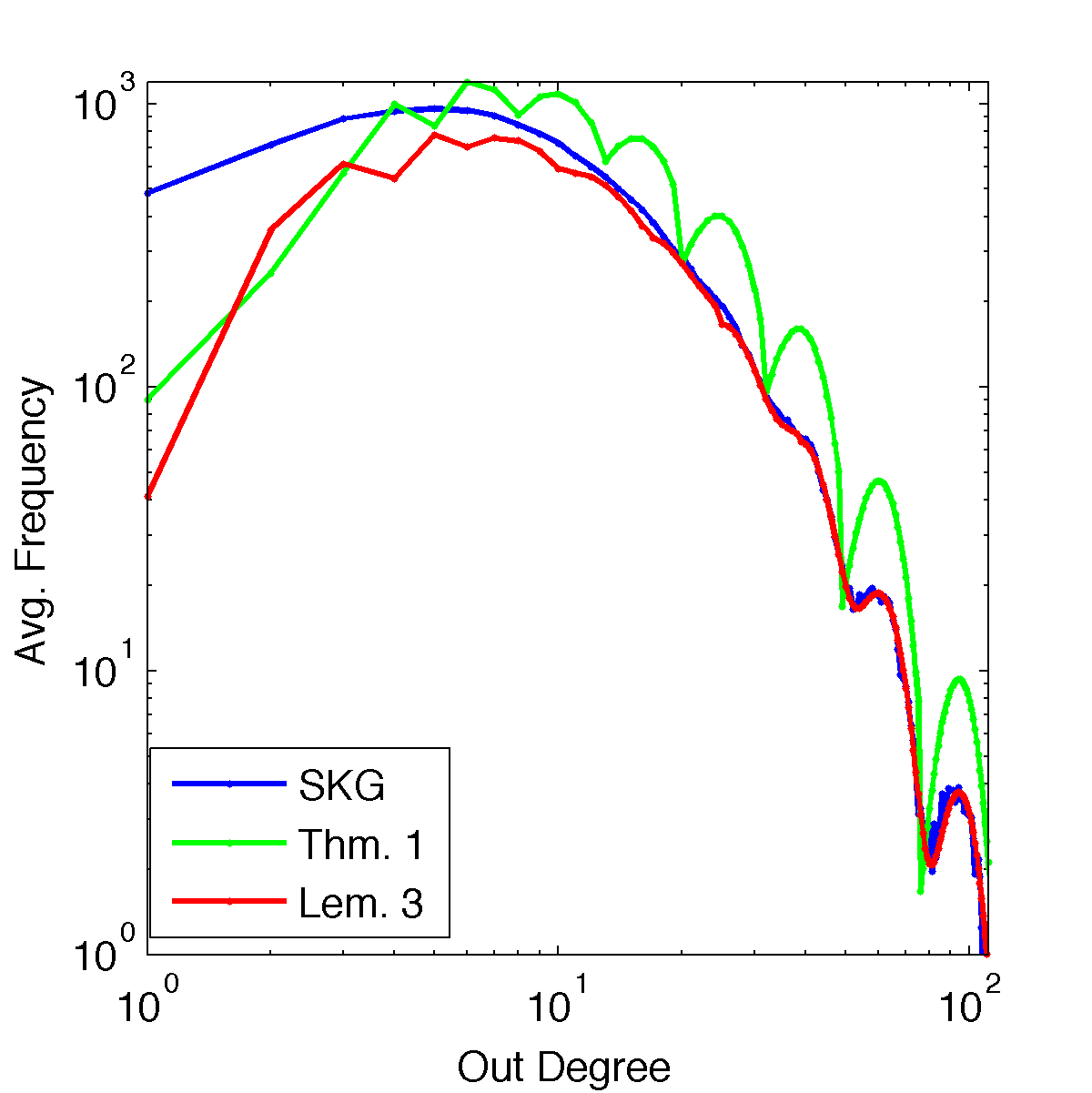}
  }
  \subfloat[WEBNotreDame (with $\ell=18$)]{
  \includegraphics[width=.45\textwidth]{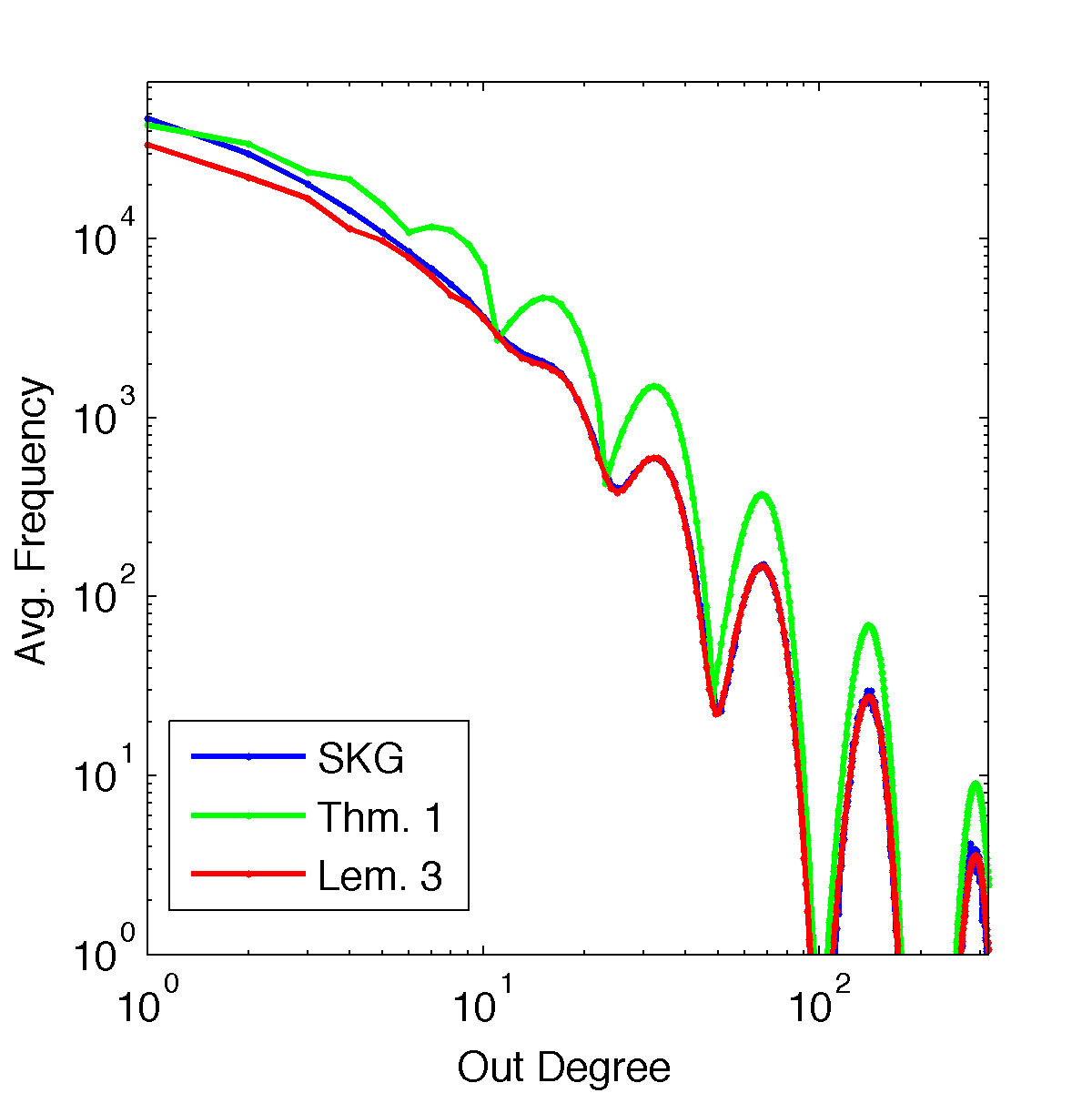}
  }

  \subfloat[Graph500 (with $\ell=16$)]{
  \includegraphics[width=.45\textwidth]{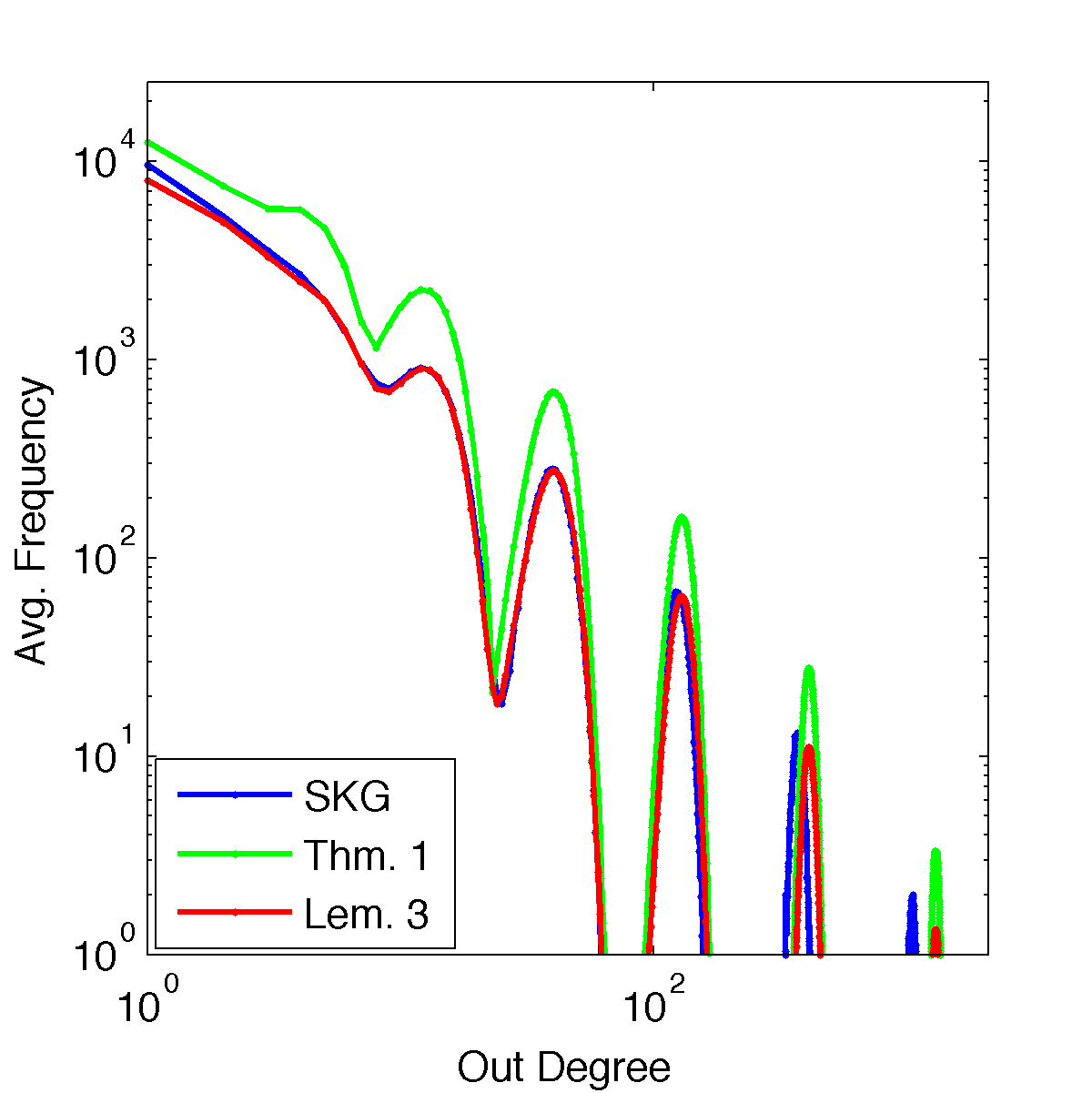}
  }
  \caption{We plot the degree distribution of graphs generated using our three different SKG
  parameter sets. %
  We then plot the respective bounds predicted by \Thm{deg-d-sum} and \Lem{deg-d}.
  Observe how \Thm{deg-d-sum} correctly guesses the peaks and troughs of the degree distribution.
  \Lem{deg-d} is practically an exact match (except when the degree
  is below $2\ell$ or, in Graph500, slight inaccuracies when the degree is too large).
  }
  \label{fig:degdist}
\end{figure*}

We begin by stating and explaining the main result of this section.
To provide clean expressions, we make certain approximations
which are slightly off for certain regions of $d$ and $\ell$ (essentially,
when $d$ is either too small or too large). Our main technical result is \Lem{deg-d},
which gives a tight expression for the degree distribution. A more
interpretable version is given first as \Thm{deg-d-sum}, which is stated 
as an upper bound. The remainder of the section gives a proof for this,
which can be skipped if the reader is only interested in the results.
This theorem expresses the oscillations between the lognormal
and exponential tail.
The lower order error terms in all the following are extremely small.

We focus on outdegrees, but these theorems hold
for indegrees as well. 
To make dependences clear, we remind the reader that the ``free" variables are $T, \Delta, \ell$.
The first two are fixed to constants, and $\ell$ is increasing. Hence, the asymptotics are
over $\ell$. All other parameters are functions of these quantities.

We begin by giving a more digestible form of our main result, stated in \Thm{deg-d-sum}.
The more precise version is given in \Lem{deg-d}. A reader interested in the general message
can skip \Lem{deg-d}.

\begin{theorem} \label{thm:deg-d-sum} 
Assume $d \in [(e\ln 2)\ell , \sqrt{n}\,]$.
If $\Gamma_d \geq \ell/2$, then $\EX[X_d]$ is negligible, i.e., $o(1)$; otherwise, 
if $\Gamma_d < \ell/2$, then (up to an additive exponential tail)
\begin{displaymath} 
\EX[X_d] \leq 
\frac{1}{\sqrt{d}} \exp\left(\frac{-d\gamma^2_d\ln^2\tau}{2}\right) {\ell\choose{\ell/2 + \Gamma_d}}.
\end{displaymath}
\end{theorem}

\emph{Remark:} This means that the expected outdegree distribution
of a SKG is bounded above by a function that oscillates between a lognormal
and an exponential tail.

Note that $\Gamma_d = \nint{\ln(d/\lambda)/\ln \tau} = \Theta(\ln d)$. 
Hence $\ell\choose{\ell/2+\Gamma_d}$ can be thought of as $\ell\choose{\ell/2 + \Theta(\ln d)}$.
The function $\ell\choose{\ell/2+x}$ represents an asymptotically normal distribution of $x$, and therefore $\ell\choose{\ell/2 + \Gamma_d}$ is 
a \emph{lognormal distribution of $d$}. This lognormal term is multiplied by $\exp(-d\gamma^2_d \ln^2\tau/2)$.
By definition, $\gamma_d \in [0,1/2]$. When $\gamma_d$ is close to $0$,
then the exponential term is almost $1$. Hence the product represents a lognormal tail.
On the other hand, when $\gamma_d$ is a constant (say $> 0.2$), then the product
becomes an exponential tail. Observe that $\gamma_d$ oscillates between $0$ and
$1/2$, leading to the characteristic behavior of SKG. As $\theta_d$ becomes
closer to an integer, there are more vertices of degree $d$. As it starts
to have a larger fractional part, the number of  vertices of degree $d$ is bounded above by an exponential tail. 
Note that there are many values of $d$ (a constant fraction) where $\gamma_d > 0.2$.
Hence, for all these $d$, the degrees are bounded above by an exponential
tail. \emph{As a result, the degree distribution cannot be a power law or a lognormal.}

The estimates provided by \Thm{deg-d-sum} 
for our three different SKG parameter sets are shown in \Fig{degdist}. Note
how this simple estimate matches the oscillations of the actual degree distribution
accurately. 

We provide a more complex expression in \Lem{deg-d} that almost
completely explains the degree distribution.
\Thm{deg-d-sum} is a direct corollary of this lemma.
In the following, the expectation is over the random choice of the graph.

\begin{lemma} \label{lem:deg-d}
For SKG, assume  $d \in [(e\ln 2)\ell, \sqrt{n}\,]$.
If $r_d \geq \ell/2$, $\EX[X_d]$ is negligible; otherwise, we have
\begin{multline*}
\EX[X_d] = \frac{1 \pm o(1)}{\sqrt{2\pi d}} 
\left\{
\exp\left(\frac{-d\delta^2_d\ln^2\tau}{2}\right) {\ell\choose{\ell/2 + r_d}}\right.\\
+ \left. \exp\left(\frac{-d(1-\delta_d)^2\ln^2\tau}{2}\right) {\ell\choose{\ell/2 + r_d + 1}}
\right\}.
\end{multline*}
\end{lemma}

We plot the bound given by this lemma in \Fig{degdist}. Note how it
completely captures the behavior of the degree distribution (barring a slight
inaccuracy for larger degrees of the Graph500 graph because we start
exceeding the upper bound for $d$ in \Lem{deg-d}). \Thm{deg-d-sum} can be
derived from this lemma, as we show below.

\begin{proof} (of \Thm{deg-d-sum}) Since $\delta_d = \theta_d - \flo{\theta_d}=\theta_d-r_d$, only one of
$\delta_d$ and $(1-\delta_d)$ is at most $1/2$.  In the former case, $\Gamma_d = r_d$ and in the latter
case, $\Gamma_d = r_d + 1$. Suppose that $\Gamma_d = r_d$. Then, 
$$ \exp\left(\frac{-d(1-\delta_d)^2\ln^2\tau}{2}\right) {\ell\choose{\ell/2 + r_d + 1}} \leq 
\exp\left(\frac{-d\ln^2\tau}{8}\right) {\ell\choose{\ell/2 + r_d + 1}} $$
Note that this is a small (additive) exponential term in \Lem{deg-d}. So we just neglect it
(and drop the leading constant of $1/\sqrt{2\pi}$) to get a simple approximation. A similar
argument works when $\Gamma_d = r_d + 1$.
\end{proof}

In the next section, we prove some preliminary claims which are building
blocks in the proof of \Lem{deg-d}. Then, we give a long intuitive explanation of how we prove \Lem{deg-d}. Finally, in \Sec{deg-proof}, we give a complete proof
of \Lem{deg-d}.

\subsection{Preliminaries} \label{sec:prelims}

We will state and prove some simple and known results in our own notation. 
This will give the reader some understanding about the various \emph{slices}
of vertices, and how the degree distribution is related to these slices.
Our first claim computes the probability
that a single edge insertion creates an outedge for node $v$. The probability
depends only on the slice that $v$ is in.

\begin{claim} \label{clm:prob-pr} 
For vertex $v \in S_r$, the probability that a single edge insertion in SKG produces an out-edge at node $v$ is 
\begin{displaymath}
  p_r = \frac{(1-4\sigma^2)^{\ell/2} \tau^r}{n}.
\end{displaymath}
\end{claim}

\begin{myproof} We consider a single edge insertion. What is the probability that this leads 
to an outedge of $v$? At every level of the insertion, the edge must go into the half
corresponding to the binary representation of $v$. If the first bit of $v$ is $0$,
then the edge should drop in the top half at the first level, and this happens with probability $(1/2+\sigma)$. 
On the other hand, if this bit is $1$, then the edge should drop in the bottom half,
which happens with probability $(1/2 - \sigma)$.
By performing this argument for every level, we get that 
\begin{displaymath}
p_r  = \left(\frac{1}{2} + \sigma\right)^{\ell/2+r}\left(\frac{1}{2} - \sigma\right)^{\ell/2-r} 
= \frac{(1-4\sigma^2)^{\ell/2}}{2^\ell} \cdot \left(\frac{1/2+\sigma}{1/2-\sigma}\right)^r 
= \frac{(1-4\sigma^2)^{\ell/2} \tau^r}{n}.
\qquad\myproofend
\end{displaymath}
\end{myproof}

Our next lemma bounds the probability that a vertex $v$ at slice $r$ has degree
$d$. Before that, we separately deal with slices where $p_r$ is very large. Essentially, we
show that slices where $p_r \geq 1/\sqrt{m}$ can be ignored. This allows for simpler calculations
later on.

\begin{claim} \label{clm:large-d} Let $R$ be the set $\{r | p_r \geq 1/\sqrt{m}\}$ and
$U = \bigcup_{r \in R} S_r$.
The probability that \emph{any} vertex in $U$ has degree less than $\sqrt{m}/2$ is at most $e^{-\Omega(\sqrt{m})}$.
\end{claim}

\begin{proof} Consider a fixed $v \in U$. Let $X_i$ be the indicator random variable for the $i$th edge insertion being incident to $v$.
The $X_i$s are i.i.d. with $\EX[X_i] \geq 1/\sqrt{m}$. The out-degree of $v$ is $X = \sum_{i=1}^m X_i$
and $\EX[X] \geq \sqrt{m}$. By a multiplicative Chernoff bound (Theorem 4.2 of~\cite{MR}),
the probability that $X \leq \sqrt{m}/2$ is at most $e^{-\sqrt{m}/8}$. The proof is completed
by taking
a union bound over all vertices in $U$ and noting that $n e^{-\sqrt{m}/8} = e^{-\Omega(\sqrt{m})}$.
\end{proof}

We will set $d = o(\sqrt{n})$. 
Our formula becomes slightly inaccurate when $d$ becomes large, but as our figures show, it is not a major issue in practice. 
The previous claim implies that the expected number of vertices
in $U$ (as defined above) with degree $d$ is vanishingly small. Therefore,
we only need to focus on slices where $p_r \leq 1/\sqrt{m}$.

\begin{lemma} \label{lem:prob-d} Let $v$ be a vertex in slice $r$.
Assume that $p_r \leq 1/\sqrt{m}$ and $d = o(\sqrt{n})$. Then for SKG,
\begin{displaymath}
  \pr[\odeg{v} = d] =
   (1+o(1))\, \frac{\lambda^d}{d!}  \frac{(\tau^r)^d}{\exp(\lambda \tau^r)}. 
\end{displaymath}
\end{lemma}

\begin{proof} The probability that $v$ has outdegree $d$ is ${m\choose d} p_r^d (1-p_r)^{m-d}$.
Since $d = o(\sqrt{n})$, we have ${m \choose d} = (1 \pm o(1))m^d/d!$. 
For $x \leq 1/\sqrt{m}$ and $m' \leq m$, we can use the Taylor series approximation,
$(1-x)^{m'} = (1 \pm o(1)) e^{-xm'}$. Using \Clm{prob-pr}, we get
\begin{align*}
{m\choose d} p_r^d (1-p_r)^{m-d} & = (1 \pm o(1)) \frac{m^d}{d!}  \left(\frac{(1-4\sigma^2)^{\ell/2} \tau^r}{n}\right)^d \exp\left(- \frac{(1-4\sigma^2)^{\ell/2} \tau^r(m-d)}{n}\right)  \\
& = (1 \pm o(1))\frac{\left(\Delta (1-4\sigma^2)^{\ell/2}\right)^d \tau^{rd}}{d!}   \exp(- \Delta (1-4\sigma^2)^{\ell/2} \tau^r) \exp(\frac{d(1-4\sigma^2)^{\ell/2} \tau^r}{n}) \\
& = (1 \pm o(1))\frac{\lambda^d}{d!}  \frac{(\tau^r)^d}{\exp(\lambda \tau^r)} \exp(dp_r).
\end{align*}
Since $p_r \leq 1/\sqrt{m}$ and $d = o(\sqrt{n})$, $dp_r = o(1)$, completing
the proof.
\end{proof}

\subsection{Understanding the degree distribution} \label{sec:und-deg}

The following is a verbal explanation of our proof strategy and
captures the essence of the math.

It will be convenient to think
of the parameters having some fixed values. Let $\lambda = 1$
and $\tau = e$. (This can be achieved with a reasonable choice
of $T, \ell, \Delta$.) We begin by looking at the different slices
of vertices. Vertices in a fixed $r$-slice have an identical behavior with
respect to the degree distribution. \Lem{prob-d} uses elementary probability arguments
to argue that the probability that a vertex in slice $r$ has outdegree $d$
is roughly
\begin{eqnarray}
\label{eq:dfac} 
\pr[\odeg{v} = d]
= \frac{\exp(dr - e^r)}{d!} .
\end{eqnarray}
When $r = \Omega(\ln d)$, the numerator will be less than $1$, and the overall
probability is $O(1/d!)$. Therefore, those slices will not have many (or any)
vertices of degree $d$. If $r = O(\ln d)$, the numerator is $o(d!)$
and the probability is still (approximately) at most $1/d!$. Observe
that when $r$ is negative, then this probability is extremely small,
even for fairly small values of $d$. This shows that half of the vertices
(in slices where the number of $1$'s is more than $0$'s) 
have extremely small degrees.

It appears that the ``sweet spot" is around $r \approx \ln d$. 
Applying Taylor approximations to appropriate ranges of $r$,
it can be shown that a suitable approximation 
of the probability of a slice $r$ vertex having degree 
$d$ is roughly $\exp(-d(r-\ln d)^2)$. 
We can now show that the SKG degree distribution is bounded \emph{above}
by a lognormal tail. Only the vertices in slice $r \approx \ln d$
have a good chance of having degree $d$. This means that the expected
number of vertices of degree $d$ is at most ${\ell\choose \ell/2 + \ln d}$.
Since the latter is asymptotically normally distributed as a function of $\ln d$,
it (approximately) represents a lognormal tail. A similar
conclusion was drawn in \cite{KiLe10}, though their approach and
presentation is very different from ours.

This is where we significantly diverge. The crucial observation is that
$r$ is a \emph{discrete} variable, not a continuous one. When
$|r - \ln d| \geq 1/3$ (say), the probability of having
degree $d$ is at most $\exp(-d/9)$. That is an exponential tail, so we can safely 
assume that vertices in those slices have no vertices of degree $d$.
Refer to \Fig{gaussian}.
Since $\ln d$ is not necessarily integral, it could be
that for \emph{all} values of $r$, $|r - \ln d| \geq 1/3$. In that case,
there are (essentially) \emph{no vertices of degree $d$}. For
concreteness, suppose $\ln d = 100/3$. Then, regardless of the value of $r$,
$|r - \ln d| \geq 1/3$. And we can immediately bound the fraction
of vertices that have this degree by the exponential tail,
$\exp(-d/9)$.
When $\ln d$ is close to being integral, then 
for $r = \nint{\ln d}$, the $r$-slice (and only this slice) will contain many vertices of degree
$d$. The quantity $|\ln d - \nint{\ln d}|$ fluctuates between
$0$ and $1/2$, leading to the oscillations in the degree distribution.

\begin{figure}
  \centering
   \includegraphics[width=.45\textwidth]{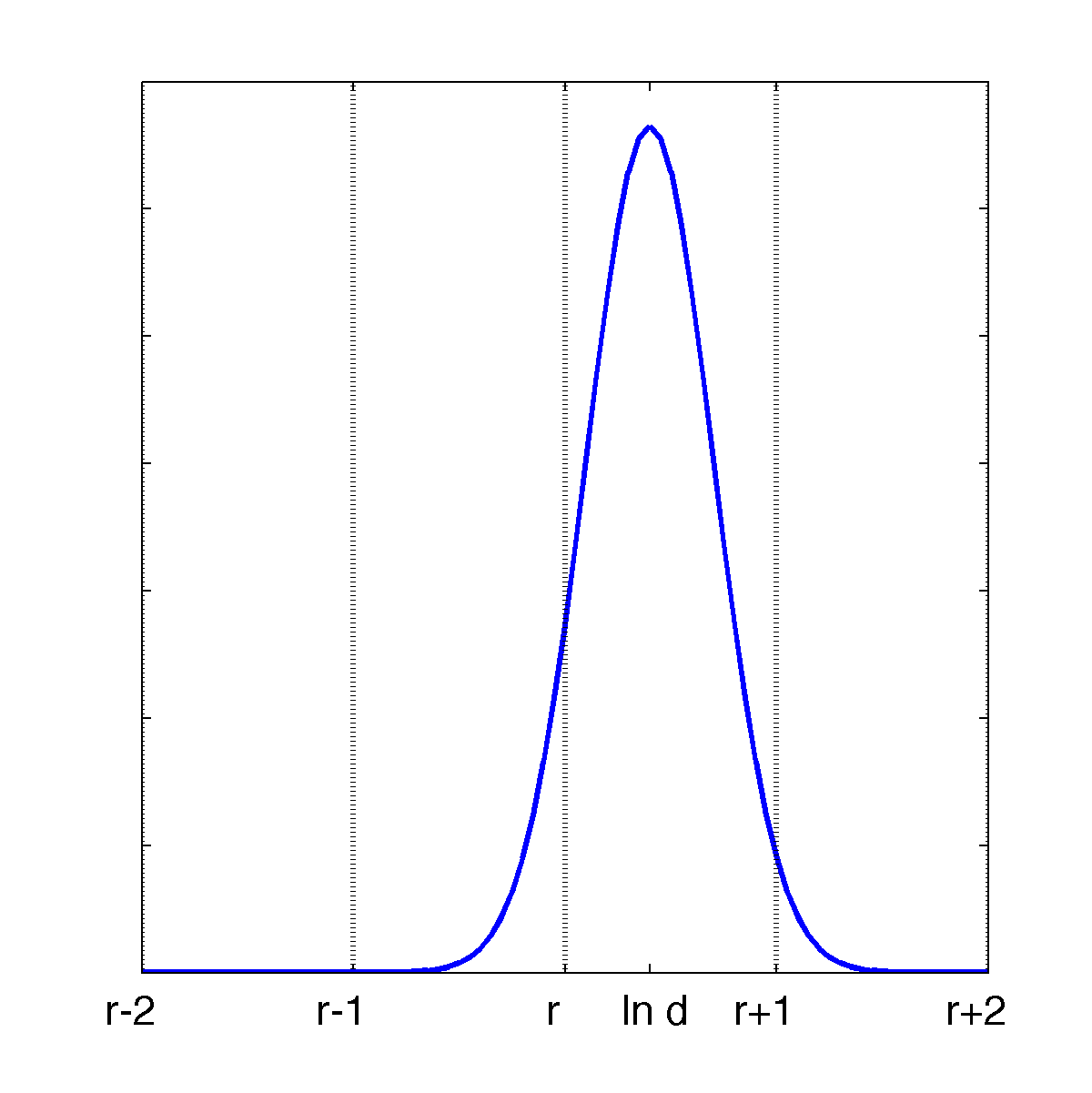}
  \caption{Probability of nodes of degree $d$ for various slices. The probability
  that a vertex of slice $r$ has degree $d$ is Gaussian distribution with a peak
  at $\ln d$. The standard deviation is extremely small. Hence, if $\ln d$ is far from integral, \emph{no slice} will have vertices of degree $d$.}
  \label{fig:gaussian}
\end{figure}

Let $\Gamma_d = \nint{\ln d}$ and $\gamma_d = |\Gamma_d - \ln d|$. 
Putting the arguments above together, we can get a very good estimate
of the number of vertices of degree $d$. 
This quantity is essentially $\exp(-\gamma^2_d d){\ell\choose \ell/2 + \Gamma_d}$,
as stated in \Thm{deg-d-sum}.
A more nuanced argument leads to the bound in \Lem{deg-d}.

\subsection{Proof of \Lem{deg-d}} \label{sec:deg-proof}

We break up the main argument into various claims. The first claim gives
an expression for the expected number of vertices of degree $d$.
This sum will appear to be a somewhat complicated sum of binomial
coefficients. But, as we later show, we can deduce that most
terms in this sum are actually negligible.

\begin{claim} \label{clm:def-g} Define $g(r) = r\ln \tau - \ln(d/\lambda)$.
Then, for SKG,
\begin{displaymath}
  \EX[X_d] = \frac{1 \pm o(1)}{\sqrt{2\pi d}}
  \sum_{r = -\ell/2}^{\ell/2} \exp\left[d(1+g(r)-e^{g(r)})\right] {\ell\choose{\ell/2+r}}.
\end{displaymath}
\end{claim}

\begin{myproof} Using \Lem{prob-d} and linearity of expectation, we can
derive a formula for $\EX[X_d]$. We then apply Stirling's approximation and the fact that $|S_r| = {\ell\choose{\ell/2+r}}$.
\begin{eqnarray*} \label{eq:xd} 
\EX[X_d] & = & (1 \pm o(1))\frac{\lambda^d}{d!} 
\sum_{r = -\ell/2}^{\ell/2}  \frac{(\tau^r)^d}{\exp(\lambda \tau^r)} |S_r| \\
& = & (1 \pm o(1))\frac{\lambda^d}{d!} 
\sum_{r = -\ell/2}^{\ell/2}  \frac{(\tau^r)^d}{\exp(\lambda \tau^r)} 
{\ell\choose{\ell/2+r}} \\
& = & \frac{1 \pm o(1)}{\sqrt{2\pi d}}\left(\frac{e\lambda}{d}\right)^d 
\sum_{r = -\ell/2}^{\ell/2}  \frac{(\tau^r)^d}{\exp(\lambda \tau^r)} 
{\ell\choose{\ell/2+r}}.
\end{eqnarray*}
Let us now focus on the quantity 
\begin{displaymath} 
\left(\frac{e\lambda}{d}\right)^d  \frac{(\tau^r)^d}{\exp(\lambda \tau^r)} 
= \exp(d+d\ln \lambda + rd\ln \tau - d\ln d - \lambda \tau^r).
\end{displaymath}
The term inside the exponent can be written as $d + d(r\ln \tau - \ln d + \ln \lambda) - d(d/\lambda)^{-1}\tau^r$.
This is $d(1+g(r)-e^{g(r)})$. Hence
\begin{displaymath}
  \EX[X_d] = \frac{1 \pm o(1)}{\sqrt{2\pi d}}\sum_{r = -\ell/2}^{\ell/2}  
  e^{d(1+g(r)-e^{g(r)})} {\ell\choose{\ell/2+r}}. \qquad \myproofend
\end{displaymath}
\end{myproof}

The key observation is that among the $\ell$ terms in the summation of \Clm{def-g},
few of them are the main contributors. All other terms sum
up to a negligible quantity. We deal with this part in the following claim.
We crucially use the assumption that $d > (e\ln 2)\ell$.
This ensures that the large slices (when $|r|$ is small) do not contribute vertices
of degree $d$.

\begin{claim} \label{clm:most} 
Let $R$ be the set of $r$ such that $|g(r)| \geq 1$. Then, for SKG, 
\begin{displaymath}
  \sum_{r \in R} \exp[d(1+g(r)-e^{g(r)})] {\ell\choose{\ell/2+r}} \leq 1.    
\end{displaymath}
\end{claim}

\begin{myproof} For convenience, define $h(r) = 1 + g(r)-e^{g(r)}$. We 
will show (shortly) that when $|g(r)| \geq 1$, $h(r) \leq -1/e$.
We assume $d > (e\ln 2)\ell$, thus
$\exp(d \cdot h(r)) \leq 2^{-\ell}$. Let $R$ be the set of all $r$ such that $|g(r)| \geq 1$.
We can easily bound the contribution of the indices in $R$ to our total sum as
\begin{displaymath}
   \sum_{r \in R} e^{dh(r)} {\ell\choose{\ell/2+r}} \leq 
   2^{-\ell} \sum_{r \in R} {\ell\choose{\ell/2+r}}
   \leq 1. 
\end{displaymath}
It remains to prove the bound on $h(r)$. 
Set $\hat h(x) = 1+x-e^x$, so $h(r) = \hat h(g(r))$. We have two cases.
\begin{itemize}
	\item $g(r) \geq 1$: Since $\hat h(x)$ is decreasing when $x \geq 1$,
	$h(r) \leq \hat h(1) = -(e-2) \leq -1/e$.
	\item $g(r) \leq -1$: Since $\hat h(x)$ is increasing for $x \leq -1$,
	$h(r) \leq \hat h(-1) = -1/e$. $\qquad \myproofend$
\end{itemize}
\end{myproof}

Now for the main technical part. The following claim with the previous ones
complete the proof of \Lem{deg-d}.

\begin{claim} \label{clm:rest} Define $R$ as in \Clm{most}.
Then, for SKG,
\begin{multline*} 
\sum_{r \notin R} \exp\left[d(1+g(r)-e^{g(r)})\right] {\ell\choose{\ell/2+r}} = \\
(1 \pm o(1)) \cdot \left\{ \exp\left(\frac{-d\delta^2_d\ln^2\tau}{2}\right) {\ell\choose{\ell/2 + r_d}} 
+ \exp\left(\frac{-d(1-\delta_d)^2\ln^2\tau}{2}\right) {\ell\choose{\ell/2 + r_d + 1}}
\right\}.
\end{multline*}
\end{claim}

\begin{myproof} Since $|g(r)| < 1$, we can perform an important approximation.
Using the expansion $e^x = 1 + x + x^2/2 + \Theta(x^3)$ for $x \in (0,1)$, 
we bound
$$ h(r) = 1 + g(r) - e^{-g(r)}  = -g(r)^2/2 + \Theta(g(r)^3)$$
We request the reader to 
pause and consider the ramifications of this approximation. The coefficient
multiplying the binomial coefficients in the sum is $\exp(-d (g(r))^2)$,
which is a Gaussian function of $g(r)$. This is what creates the Gaussian-like behavior of the probability of vertices of degree $d$ among the various slices. 
We now need to understand when $g(r)$ is close to $0$, since the corresponding
terms will provide the main contribution to our sum. So for any $d$,
some slices are ``picked out" to have expected degree $d$, whereas
others are not. This depends on what the value of $g(r)$ is.
Now on, it only requires (many) tedious calculations to get the final result.

What are the different possible values of $g(r)$? 
We remind the reader that $g(r) = r\ln \tau - \ln(d/\lambda)$.
Observe that $r_d = \lfloor \ln(d/\lambda)/\ln \tau \rfloor$
minimizes $|g(r)|$ subject to $g(r) < 0$ and $r_d + 1$ (which
is the corresponding ceiling) minimizes $|g(r)|$ subject to $g(r) \geq 0$.
For convenience, denote $r_d$ by $r_f$ (for floor) and $r_d+1$ by $r_c$
(for ceiling).

Consider some $r$ such that $|g(r)| < 1$. It is either of the form
$r = r_c + s$ or $r_f - s$, for integer $s \geq 0$. We will
sum up all the terms corresponding to the each set separately.
For convenience, denote the former set of values of $s$'s 
such that $|g(r_c + s)| < 1$ by $S_1$, and define $S_2$
with respect to $r_f - s$. This allows us to split 
the main sum into two parts, which we deal with separately.

{\bf Case 1 (the sum over $S_1$):} 
\begin{multline*} \sum_{s \in S_1} \exp\left[d(1+g(r)-e^{g(r)})\right]{\ell\choose{\ell/2 + r_c + s}} \hfill\\
 \!= \! (1 \pm o(1)) \exp(\frac{- d(g(r_c)^2)}{2}) {\ell\choose{\ell/2 + r_c}} 
+ (1 \pm o(1)) \sum_{\substack{s \in S_1\\s \neq 0}} \exp(\frac{- d(g(r_c + s)^2)}{2}) {\ell\choose{\ell/2 + r_c + s}}
\end{multline*}
We substitute $g(r_c + s) = g(r_c) + s\ln \tau $ into the second part, and show
that we can bound this whole summation as an error term. Note that  both $s$ and $\ln \tau$ are positive by construction.
\begin{multline*}
\sum_{s \in S_1, s \neq 0} \exp(- d(g(r_c + s)^2)/2) {\ell\choose{\ell/2 + r_c + s}} \\
  \begin{aligned}
    & \leq \sum_{s \in S_1, s \neq 0} \exp[- d(g(r_c)^2 + s^2(\ln \tau)^2)/2] {\ell\choose{\ell/2 + r_c + s}} \\
    & \leq \exp(- d(g(r_c)^2)/2)  \sum_{s > 0} \exp(-ds^2(\ln \tau)^2/2) {\ell\choose{\ell/2 + r_c + s}} \\
    & = o\left(\exp(-d(g(r_c)^2)/2) {\ell\choose{\ell/2 + r_c}}\right).
  \end{aligned}
\end{multline*}
For the last inequality, observe that ${\ell \choose \ell/2 + r_c + s} \leq \ell^s {\ell \choose \ell/2+r_c}$. Since $d \geq \ell$, the exponential decay of $\exp(\Theta(-ds^2))$ completely kills this summation.

{\bf Case 2 (the sum over $S_2$):}
Now, we apply an identical argument for $r = r_f - s$. We have $g(r) = g(r_f) - s\ln \tau$. 
Applying the same calculations as above, 
$$  \sum_{s \in S_2} \exp\left[d(1+g(r)-e^{g(r)})\right] {\ell\choose{\ell/2 + r_f + s}} 
= (1 \pm o(1)) \exp(- d(g(r_f)^2)/2) {\ell\choose{\ell/2 + r_f}} $$

Adding the bounds from both the cases, we conclude
\begin{eqnarray} 
& & \sum_{r \notin R} \exp\left[d(1+g(r)-e^{g(r)})\right] {\ell\choose{\ell/2+r}}  \nonumber \\
& & = (1 \pm o(1)) \cdot \left\{ \exp(-dg(r_f)^2/2) {\ell\choose{\ell/2 + r_f}} + \exp(-dg(r_c)^2/2) {\ell\choose{\ell/2 + r_c}} \label{eq:fin}
\right\} 
\end{eqnarray}
We showed earlier that $r_f = r_d$ and $r_c = r_d + 1$. We remind the reader that $\theta_d = \ln(d/\lambda)/\ln \tau$, $r_d = \flo{\theta_d}$,
and $\delta_d = \theta_d - r_d$. Hence $g(r_f) = g(\theta_d) - \delta_d \ln \tau = -\delta_d \ln \tau$.
Since $r_c = r_f + 1$, $g(r_c) = \ln \tau + g(r_f) = (1-\delta_d)\ln \tau$. We substitute in \Eqn{fin}
to complete the proof. 
\qquad \myproofend
\end{myproof}

\section{Enhancing SKG with Noise: NSKG} \label{sec:nskg}

Let us now focus on a noisy version of SKG that removes the fluctuations
in the degree distribution.  We will refer to  our proposed noisy SKG model as  NSKG. The idea is quite simple. For each level $i \leq \ell$, define
a new matrix $T_i$ in such a way that the expectation
of $T_i$ is just $T$.
At level $i$ in the edge insertion, we use the matrix $T_i$ to choose the appropriate quadrant.

Here is a formal description. For convenience, we will assume that $T$ is symmetric.
It is fairly easy to generalize to general $T$.
Let $b$ be our noise parameter such that $b \leq \min((t_1+t_4)/2, t_2)$. For level $i$, choose $\mu_i$
to be a uniform random number in the range $[-b,+b]$. Set $T_i$ to be 
\begin{displaymath}
  \M{T}_i =
  \begin{bmatrix} 
    t_1 - \frac{2\mu_i t_1}{t_1+t_4} & t_2 + \mu_i \\ 
    t_3 + \mu_i & t_4 - \frac{2\mu_i t_4}{t_1+t_4}
  \end{bmatrix}
\end{displaymath}

Note that $T_i$ is symmetric, its entries sum to $1$, and all entries are positive.
This is by no means the only model of noise, but it is certainly
convenient for analysis. Each level involves only one random number $\mu_i$,
which changes all the entries of $T$ in a linear fashion. Hence,
we only need $\ell$ random numbers in total. For convenience,
we list out the noise parameters of NSKG in \Tab{paramsnoise}.

\begin{table}[tbp]
  \caption{Parameters for NSKG}
  \label{tab:paramsnoise}
  \centering
  \begin{boxedminipage}{.9\textwidth}
    \begin{itemize}
    \item $b = $ noise parameter $\leq \min((t_1+t_4)/2, t_2)$
    \item $\mu_i = $ noise at level $i = 1,\dots,\ell$
    \item $T_i = \begin{bmatrix} 
        t_1 - \frac{2\mu_i t_1}{t_1+t_4} & t_2 + \mu_i \\ 
        t_3 + \mu_i & t_4 - \frac{2\mu_i t_4}{t_1+t_4}
      \end{bmatrix} = $ noisy generating matrix at level $i = 1,\dots,\ell$
    \end{itemize}
  \end{boxedminipage}
\end{table}

In Figures \ref{fig:noisy_degdist_graph500}, \ref{fig:noisy_degdist_CA-HEP-PH},
and \ref{fig:noisy_degdist_WEB-NOTREDAME}, we show the effects of noise.
Observe how even a noise parameter as small as $0.05$ (which is extremely small
compared to the matrix values) significantly reduces the magnitude
of oscillations. A noise of $0.1$ almost removes the oscillations. 
(Even this noise is very small, since the standard deviation
of this noise parameter is at most $0.06$.)
Our proposed method of adding noise 
dampens the undesirable exponential tail behavior of SKG, leading
to a monotonic degree distribution.

\begin{figure}[htb]
  \centering
  \subfloat[CAHepPh]{\label{fig:noisy_degdist_CA-HEP-PH}
  \includegraphics[width=.48\columnwidth,trim=0 0 5 5,clip]{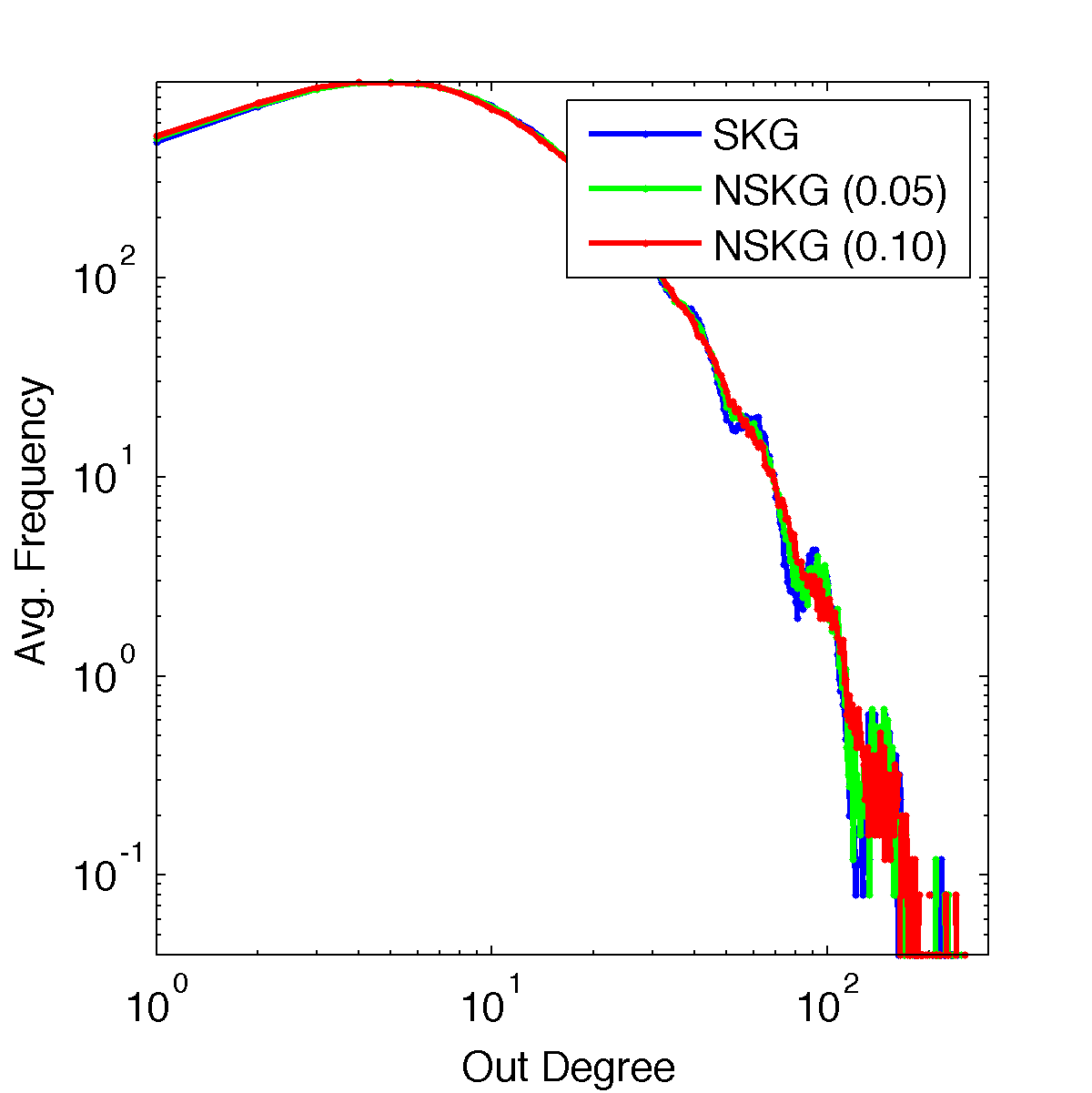}
  }
  \subfloat[WEBNotreDame]{\label{fig:noisy_degdist_WEB-NOTREDAME}
  \includegraphics[width=.48\columnwidth,trim=0 0 5 5,clip]{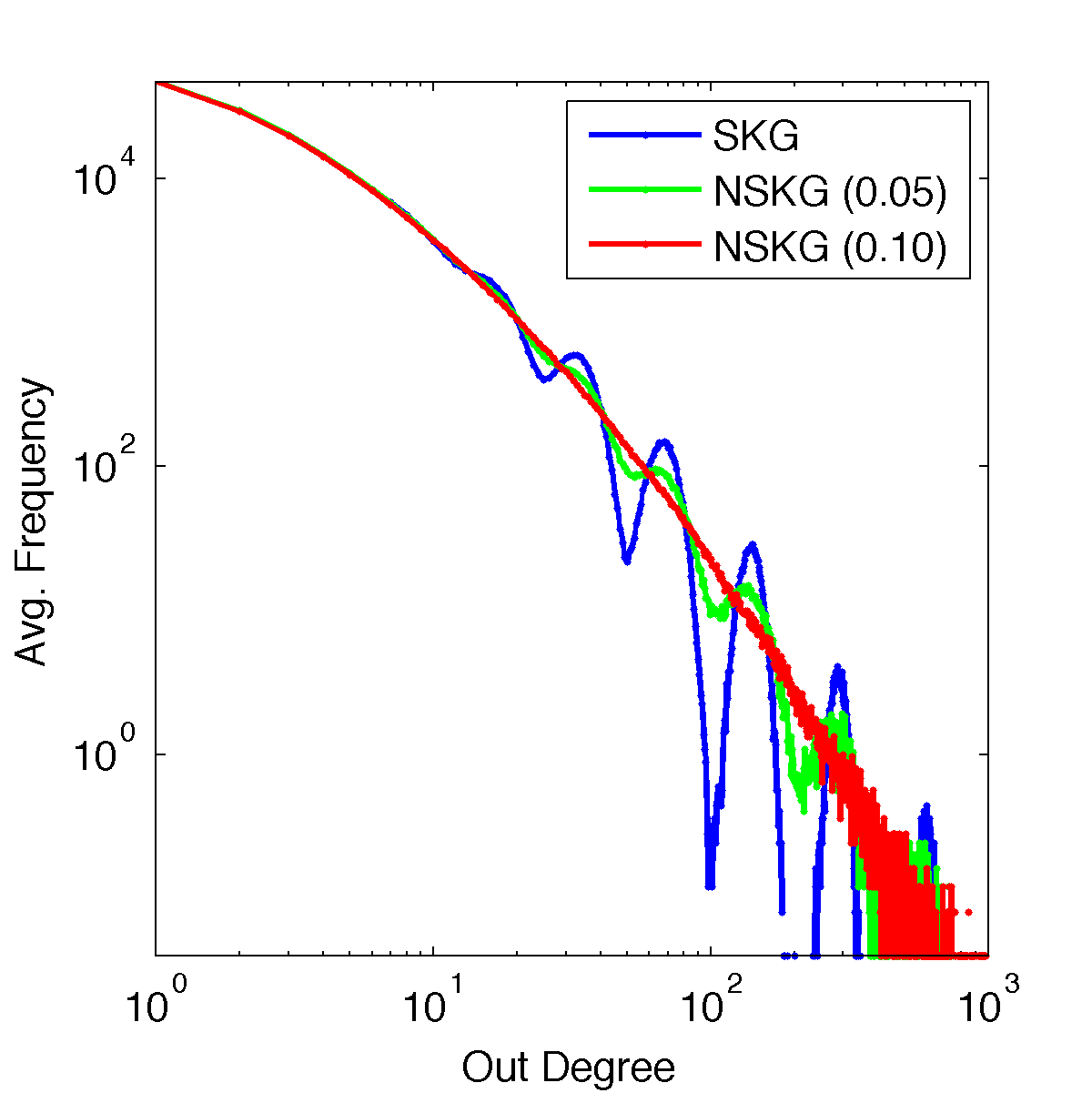}
  }
  \caption{
  The figures show the degree distribution of standard SKG and NSKG as the averages of 25 instances. Notice
  how effectively a noise of $0.1$ straightens the degree distribution.}
  \label{fig:degdist_noisy}
\end{figure}

\par~
\subsection{Why does noise help?} \label{sec-und-noise}

Before we state our formal theorem, let us set some asymptotic notation that will
allow for a more readable theorem. 
We will use the $O(\cdot)$ notation to suppress constant factors, where (for notational convenience) these constants
\emph{may} depend on the constants in the matrix $T$. 
As before, $o(1)$ is a quantity that goes to zero as $\ell$ grows.

Our formal theorem says that when the noise is ``large enough,"  we can show that the degree distribution
has at least a lognormal tail on average. This is a significant change from SKG,
where many degrees are below an exponential tail.

\begin{theorem} \label{thm:lognormal} Let noise $b$ be set to $c/\sqrt{\ell}$ for positive $c$, such that $c/\sqrt{\ell} < \min((t_1+t_4)/2, t_2)$.
Then the expected degree distribution for NSKG is bounded below by a lognormal. Formally,
when $\Gamma_d \leq \ell/2$ and $d \leq \sqrt{n}$,
\begin{displaymath}
  \EX[X_d] \geq \frac{\nu(c)}{d} {\ell\choose \ell/2+\Gamma_d}.
\end{displaymath}
Here $\nu(c)$ is some positive function of $c$. (This is independent of $\ell$, so for constant $c$, $\nu(c)$ is a positive constant.)
\end{theorem}

This bound tells us that as $\ell$ increases,
we need \emph{less} noise to get a lognormal tail. 
From a Graph 500 perspective, if we determine (through experimentation) 
that for some small $\ell$ a certain amount of noise suffices, the \emph{same}
amount of noise is certainly enough for \emph{larger $\ell$}.

We now provide a verbal description of the main ideas.
Let us assume that $\lambda = 1$ and $\tau = e$,
as before. We focus our attention on a vertex $v$ of slice $r$,
and wish to compute the probability that it has degree $d$.
Note the two sources of randomness: one coming from the choice
of the noisy SKG matrices, and the second from
the actual graph generation. We associate a \emph{bias parameter}
$\rho_v$ with every vertex $v$. This can be thought of as some
measure of how far the degree behavior of $v$ deviates from
its noiseless version. Actually, it is the random variable $\ln \rho_v$ 
that we are interested in. \full{Intuitively, 
this can just be thought of as a Gaussian random variable with mean zero.}
\confer{It can be shown that $\ln \rho_v$ is distributed like a Gaussian.}
The \emph{distribution} of $\rho_v$ is identical for all vertices in 
slice $r$. (Though it does not matter for our purposes,
for a given instantiation of the noisy SKG matrices,
vertices in the same slice can have different biases.)

We approximate the probability that $v$ has degree $d$ by \full{(refer to \Clm{nprob-d}) }
\begin{displaymath}
  \pr[\odeg{v}=d] = {\exp(dr + d\ln \rho_v - \rho_v e^r)}/{d!}. 
\end{displaymath}
After some simplifications, this is roughly equal to
$\exp(-d(r-\ln d -\ln \rho_v)^2)$. The additional
$\ln \rho_v$ will act as a \emph{smoothing} term.
Observe that even if $\ln d$ has a large fractional part, we could
still get vertices of degree $d$. Suppose $\ln d = 10.5$,
but $\ln \rho_v$ happened to be close to $0.5$. Then
vertices in slice $\nint{\ln d}$ would have degree $d$
with some nontrivial probability. Contrast this with regular SKG,
where there is almost no chance that degree $d$ vertices exist.

Think of the probability as $\exp(d(r-\ln d - X)^2)$, where
$X$ is a random variable. The expected probability will be
an average over the distribution of $X$. Intuitively, instead
of the probability just being $\exp(d(r-\ln d)^2)$ (in the case of SKG),
it is now the \emph{average} value over some interval. If the standard
deviation of $X$ is sufficiently large, even though $\exp(d(r-\ln d)^2)$
is small, the average of $\exp(d(r-\ln d - X)^2)$ can be large.
Refer to \Fig{gaussian-shaded}.

We know that $X$ is a Gaussian random variable (with some standard deviation
$\sigma$). So we can formally express the (expected) probability that $v$
has degree $d$ as an integral,
\begin{displaymath}
P(\odeg{v}=d \;\vert\; \tau = e, \; \lambda = 1) = 
 \int_{-\infty}^{+\infty} \exp(d(r-\ln d - X)^2) \cdot e^{-X^2/2\sigma^2} dX. 
\end{displaymath}
This definite integral can be evaluated exactly (since it is just a Gaussian). Intuitively, this is roughly
the average value of $\exp(d(r-\ln d - X)^2)$, where $X$ ranges from $-\sigma$
to $+\sigma$. Suppose $\sigma > 1$. Since $r$ ranges over the integers, 
there is always \emph{some} $r$ such that $|r - \ln d| < 1$. For this value
of $r$, the average of $\exp(d(r-\ln d - X)^2)$ over the range $X \in [-1,+1]$
will have a reasonably large value. This ensures that (in expectation) many 
vertices in this slice $r$ have degree $d$. This can be shown
for all degrees $d$, and we can prove that the degree distribution
is at least lognormal. 

This is an intuitive sketch of the proof. The random variable $\ln \rho_v$ is not exactly
Gaussian, and hence we have to account for errors in such an approximation. We do not
finally get a definite integral that can be evaluated exactly, but we can give
good bounds for its value.

\begin{figure}
  \centering
   \includegraphics[width=.45\textwidth]{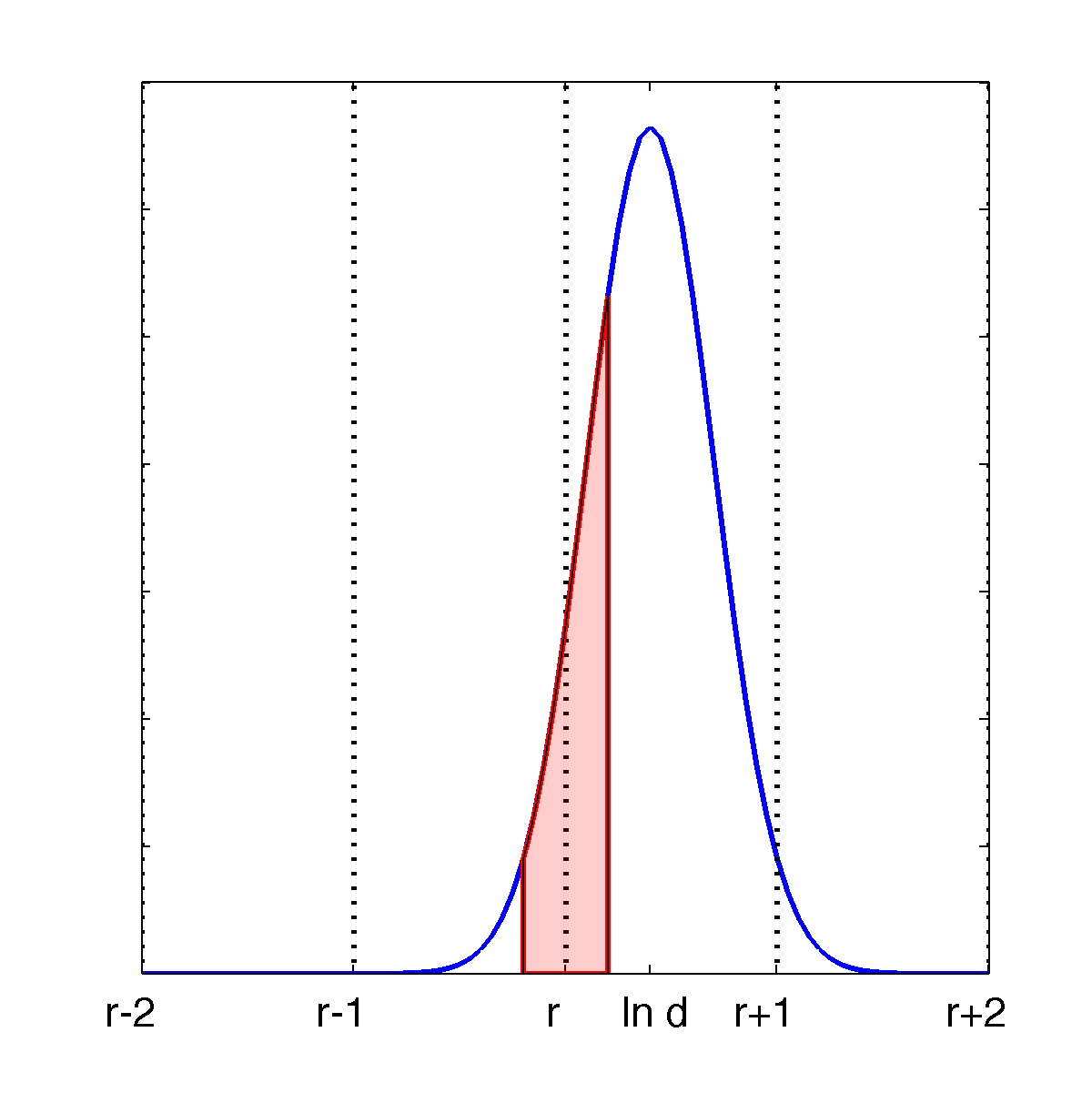}
  \caption{The effect of noise. The underlying Gaussian curve is the same as one
  in \Fig{gaussian}. Adding noise can be thought of as an average over the Gaussian. 
  So the probability that a vertex from slice $r$ has degree is the area of the shaded
  region.}
  \label{fig:gaussian-shaded}
\end{figure}

\subsection{Preliminaries for analysis} \label{sec:app-noisy}

There are many new parameters we need to introduce for our NSKG analysis. Each
of these quantities is a random variable that depends on the choice
of the matrices $T_i$. We list them below. 
\begin{itemize} 
		\item $\sigma_i = t_1 - \frac{2\mu_i t_1}{t_1+t_4} + t_2 + \mu_i - 0.5
			= \sigma + \mu_i(1-\frac{2t_1}{t_1+t_4})$.
		\item $\alpha_i = (1/2+\sigma_i)/(1/2+\sigma)$. It will be convenient to express this
		in terms of $\mu_i$, replacing the dependence on $\sigma_i$.
		$$ \alpha_i = (1/2+\sigma_i)/(t_1 + t_2) = 1 - \mu_i \frac{(t_1 - t_4)}{(t_1 + t_2)(t_1 + t_4)} $$
		\item $\beta_i = (1/2-\sigma_i)/(1/2-\sigma)$. Performing a calculation similar
		to the one above,
		$$ \beta_i = (1/2-\sigma_i)/(t_3 + t_4) = 1 + \mu_i \frac{(t_1 - t_4)}{(t_3 + t_4)(t_1 + t_4)} $$
		\item $b_\alpha, b_\beta$: We set
		$$ b_\alpha = \frac{b(t_1-t_4)}{(t_1+t_2)(t_1+t_4)} = \frac{4b\sigma}{(1+2\sigma)(t_1+t_4)} $$
		Similarly,
		$$ b_\beta = \frac{b(t_1-t_4)}{(t_3+t_4)(t_1+t_4)} = \frac{4b\sigma}{(1-2\sigma)(t_1+t_4)}$$
		Hence, $\alpha_i$ is distributed
		uniformly at random in $[1-b_\alpha,1+b_\alpha]$, and $\beta_i$ is uniformly
		random in $[1-b_\beta,1+b_\beta]$. Note that $b_\alpha, b_\beta = \Theta(c/\sqrt{\ell})$.
		\item $\rho_v$: Let $v$ be represented as a bit vector $(z_1,\ldots,z_k)$. The \emph{bias} for $v$ 
			is $\rho_v = \prod_{i:z_i=0} \alpha_i \prod_{i:z_i=1} \beta_i$. We set
			$\lambda_v = \lambda \rho_v$.
\end{itemize}

\subsection{The behavior of $\ln \rho_v$} \label{sec:rhov}

We need to bound the behavior of $\ln \rho_v$, which is 
$\sum_{i:z_i=0} \ln \alpha_i + \sum_{i:z_i=1} \ln \beta_i$. Observe that
this is a sum of independent random variables. By the Central Limit Theorem,
we expect $\ln \rho_v$ to be distributed as a Gaussian, but we still need to investigate the variance of this distribution. 
Approximately (since $b_\alpha$ and  $b_\beta$ are small), $\ln \alpha_i$
is uniformly random in $[-b_\alpha,b_\alpha]$, so the variance of $\ln \alpha_i$
is $\Theta(b^2_\alpha) = \Theta(1/\ell)$. A similar statement holds for $\ln \beta_i$,
and we bound the variance of $\ln \rho_v$ by $\Theta(1)$.
So the probability density function (pdf) of $\ln \rho_v$ is roughly concentrated in a constant-sized interval
of size $1$ (around $0$). This is what we will formally show in this section. 
We will need a pointwise convergence guarantee for the pdf
of $\ln \rho_v$. 
Throughout this section, we will use various functions of the form $\nu_1(c), \nu_2(c), \ldots$. These are
strictly positive constant functions of $c$ (for $c > 0$), and are a convenient way of tracking dependences on $c$.
The reader should interpret $\nu_a(c)$ to be some constant that depends on $c$ (and $T$ and $\Delta$, which are fixed),
but is independent of $\ell$.
The main lemma of this section is the following.

\begin{lemma} \label{lem:rho-v} Set $\htau = \max(\ln \tau, 2)$. Let $f_v(x)$ be the pdf of $\ln \rho_v$.
For $|x| \leq \htau$, $f_v(x) \geq \nuone(c)$.
\end{lemma}

We will first prove  \Lem{rho-v} as a  direct result of two claims stated below. 
Then we  will  prove these claims  in the subsequent subsections. The first claim, the more technical of the two, shows
that $\ln \rho_v$ has a sufficiently large probability of attaining a constant value.

\begin{claim} \label{clm:prob-int} There exists a constant $C > \htau$,
such that the probability that $\ln \rho_v$ lies in $[\htau,C]$ is at least $\nutwo(c)$ and that
of lying in $[-C,-\htau]$ is also at least $\nutwo(c)$.
\end{claim}

The next claim will be a consequence of the unimodularity of $f_v(x)$.

\begin{claim} \label{clm:min} For any $x \in [x_1, x_2]$, $f_v(x) \geq \min(f_v(x_1),f_v(x_2))$.
\end{claim}

Now for the proof of \Lem{rho-v}.

\begin{proof} (of \Lem{rho-v}) By \Clm{prob-int}, the probability that $\ln \rho_v$ lies in $I := [-C,-\htau]$ is at least $\nutwo(c)$.
Therefore, $(C-\htau)\max_{x \in I} f_v(x) \geq \nutwo(c)$. Suppose the maximum is achieved at $x_1$.
This means that there exists $x_1 \in [-C,-\htau]$, $f_v(x_1) = \Omega(\nutwo(c))$. Similarly, there exists some $x_2 \in [\htau,C]$ such that
$f_v(x_2) = \Omega(\nutwo(c))$. Observe that for any $x$ such that $|x| \leq \htau$,
$x \in [x_1, x_2]$. By \Clm{min}, for any such $x$, $f_v(x) =  \Omega(\nutwo(c))$.
Therefore, we can bound $f_v(x) \geq \nuone(c)$, for some positive function $\nuone$.
\end{proof}

\subsubsection{Proving \Clm{prob-int}} \label{sec:proof-prob} We begin
with notational setup. We fix some vertex $v$.
For convenience, define the variables $\halpha_i$ (for all $i \leq \ell$).
If $z_i = 0$, set $\halpha_i = \alpha_i$
and $\halpha_i = \beta_i$ otherwise. We can write $\ln \rho_v = \sum_i \ln \halpha_i$.
The random variable $\halpha_i$ is uniform in $[1-b_i,1+b_i]$, where $b_i$ is either
$b_\alpha$ or $b_\beta$ appropriately. Set the zero mean random variable $X_i = \ln \halpha_i - \EX[\ln \halpha_i]$. We have 
the following series of facts.

\begin{claim} \label{clm:basic-xi} 
\begin{asparaitem}
	\item The pdf of $\ln \halpha_i$, denoted by $h_i(x)$, is given as follows. For $x \in [\ln(1-b_i),\ln(1+b_i)]$, $h_i(x) = e^x/2b_i$,
	and zero otherwise.
	\item $|\EX[\ln \halpha_i]| = O(c^2/\ell)$, $\EX[X^2_i] = \Theta(c^2/\ell)$, and $\EX[|X_i|^3] = O(c\EX[X^2_i]/\sqrt{\ell})$.
\end{asparaitem}
\end{claim}

\begin{proof} 
The pdf of $\halpha_i$ is $h_\alpha(x) = 1/2b_\alpha$ for $x \in [1-b_\alpha,1+b_\alpha]$ and zero otherwise. 
For any monotone function $F(x)$, the pdf of $F(\halpha_i)$ is given by $|dF^{-1}(x)/dx|h(x)$.
Setting $F$ as the function $\ln$, the pdf of $\ln \alpha_i$, $h_i(x)$, is given
by $e^x/2b_\alpha$ for $x \in [\ln(1-b_\alpha),\ln(1+b_\alpha)]$ and zero otherwise.
$$ \EX[\ln \halpha_i] = \int^{\ln(1+b_i)}_{\ln(1-b_i)} xh_i(x) dx = (2b_i)^{-1} \int^{\ln(1+b_i)}_{\ln(1-b_i)} xe^x dx $$
Using integration by parts,
\begin{eqnarray*}
\int^{\ln(1+b_i)}_{\ln(1-b_i)} xe^x dx & = & [xe^x]\Big|_{\ln(1-b_i)}^{\ln(1+b_i)} - \int^{\ln(1+b_i)}_{\ln(1-b_i)} e^x dx \\
& = & [(1+b_i)\ln(1+b_i) - (1-b_i)\ln(1-b_i)] - [(1+b_i) - (1-b_i)] \\
& = & b_i \ln(1-b^2_i) + \ln(1+b_i) - \ln(1-b_i) - 2b_i
\end{eqnarray*}
Taking absolute values,
$$\Big| \int^{\ln(1+b_i)}_{\ln(1-b_i)} xe^x dx \Big| \leq | b_i \ln(1-b^2_i)| + |\ln(1+b_i) - \ln(1-b_i) - 2b_i|$$
The first term is at most $2b^3_i$. For the second term, we need a finer Taylor approximation.
\begin{eqnarray*} \ln(1+b_i) - \ln(1-b_i) - 2b_i \leq (b_i - b^2_i/2 + b^3_i) - (-b_i - b^2_i/2) - 2b_i \leq b^3_i \\
\ln(1+b_i) - \ln(1-b_i) - 2b_i \geq (b_i - b^2_i/2) - (-b_i - b^2_i/2 - b^3_i) - 2b_i \geq -b^3_i \\
\end{eqnarray*}
All in all, $|\EX[\ln \halpha_i]| \leq O(b^2_i) = O(c^2/\ell)$.
$$ \EX[X^2_i] = \EX[(\ln \halpha_i)^2] - (\EX[\ln \halpha_i])^2 $$
$$ \EX[(\ln \halpha_i)^2] = (2b_i)^{-1} \int^{\ln(1+b_i)}_{\ln(1-b_i)} x^2e^{x} dx $$
To get an upper bound for this term, we use the following inequalities: $\ln(1+b_\alpha) \leq 2b_\alpha$, $\ln(1-b_\alpha) \geq -2b_\alpha$,
$e^{x} \leq e$. That gives $\EX[(\ln \halpha_i)^2] \leq e(2b_i)^{-1} \int^{2b_i}_{-2b_i} x^2dx$ $=O(b^2_i)$.
For a lower bound, we use: $\ln(1+b_\alpha) \geq b_\alpha/2$, $\ln(1-b_\alpha) \leq -b_\alpha/2$,
$e^{x} \geq 1/e$. Hence, $\EX[(\ln \halpha_i)^2] \geq (2eb_i)^{-1} \int^{b_i/2}_{-b_i/2} x^2dx$ $ = \Omega(b^2_i)$.
Note that $(\EX[\ln \halpha_i])^2 \leq b^4_i$, which is much small than $b^2_i$ for sufficiently small $b_i$.
We conclude that $\EX[X^2_i] = \Theta(b^2_i) = \Theta(c^2/\ell)$.

For the final bound, we use a trivial estimate. We have $\EX[|X_i|^3] \leq \max(|X_i|)\EX[X^2_i] \leq 2b_i\EX[X^2_i]$.
\end{proof}

We now state the Berry-Esseen Theorem \cite{Ber41,Ess42}, a crucial ingredient of our proof. 
This theorem bounds the convergence rate of a sum of independent random variables to a
Gaussian.

\begin{theorem} \label{thm:be} [Berry-Esseen] Let $X_1, X_2, \ldots, X_\ell$ be independent
random variables with $\EX[X_i] = 0$, $\EX[X^2_i] = \xi^2_i$, and $\EX[|X_i|^3] = \iota_i < \infty$.
Let $S$ be the sum $\sum_i X_i / \sqrt{\sum_i \xi^2_i}$. Let $F(x)$ denote the cumulative distribution function (cdf) of $S$
and $\Phi(x)$ be the cdf of the standard normal (the pdf is $(2\pi)^{-1/2} e^{-x^2/2}$). 
Then, for an absolute constant $C_1 > 0$,
$$ \sup_x |F(x) - \Phi(x)| \leq C_1 \Big(\sum_i \xi^2_i\Big)^{-3/2} \sum_i \iota_i. $$
\end{theorem}

\begin{proof} (of \Clm{prob-int}) We set $X = \sum_i X_i$ $= (\ln \rho_v - \EX[\ln \rho_v])/\sqrt{\sum_i \EX[X^2_i]}$. 
By \Clm{basic-xi},
$|\EX[\ln \rho_v]| = |\sum_i \EX[\ln \halpha_i]|$ $\leq \sum_i |\EX[\ln \halpha_i]| = O(c^2)$ and $\sum_i \EX[X^2_i] = \Theta(c^2)$. 
Note that $X$ is just an increasing linear function of $\ln \rho_v$. 
Set function $r(x) = (x - \EX[\ln \rho_v])/\sqrt{\sum_i \EX[X^2_i]}$,
so $X = r(\ln \rho_v)$.
For any interval $I = [x_1, x_2]$,
$\Pr[\ln \rho_v \in I] = \Pr[X \in [r(x_1),r(x_2)]]$. 
Since $|r(\htau)|$ is some
constant function of $c$, we can find a constant $C$ such the $r(C)$
is strictly larger than $|r(\htau)|$.
Setting $y_1 = r(\htau)$, $y_2 = r(C)$ and using the notation from \Thm{be},

\begin{eqnarray*} \Pr[X \in [y_1, y_2]] & = & F(y_2) - F(y_1) = \Phi(y_2) - \Phi(y_1) + (F(y_2) - \Phi(y_2)) + (\Phi(y_1) - F(y_1)) \\
& \geq & \Phi(y_2) - \Phi(y_1) - |F(y_2) - \Phi(y_2)| - |F(y_1) - \Phi(y_1)|.
\end{eqnarray*}
Since $y_1 < y_2$ and are constant functions of $c$,
$\Phi(y_2) - \Phi(y_1) \geq \nuthr(c)$. By the Berry-Esseen theorem (\Thm{be}), $|F(x_2) - \Phi(x_2)| + |F(x_1) - \Phi(x_1)|
\leq 2C_1(\sum_i \xi^2_i)^{-3/2}\sum_i \iota_i$. By \Clm{basic-xi} $\iota_i = O(c\xi^2_i/\sqrt{\ell})$
and $\sum_i \xi^2_i = \Theta(c^2)$. So the Berry-Esseen bound is at most $2C_1c(\sum_i \ell\xi^2_i)^{-1/2} = O(1/\sqrt{\ell})$. By setting $C$ to be a large enough constant, 
we can ensure that $\Phi(y_2) - \Phi(y_1) > 2C_1c(\sum_i \ell\xi^2_i)^{-1/2}$.

We deduce that $\Pr[X \in [x_1, x_2]] \geq \nutwo(c)$, for some positive function $\nutwo$. A similar proof holds for $[-C,-\htau]$.
\end{proof} 

\subsubsection{Proving \Clm{min}} 
We state some technical
definitions and results about convolutions of unimodal functions.

\begin{definition} A pdf $f(x)$ is \emph{unimodal} if there exists an $a \in \Real$ such
that $f$ is non-decreasing on $(-\infty,a)$ and non-increasing on $(a,\infty)$.

A pdf $f(x)$ is \emph{log-concave} if $Q := \{x: f(x) > 0\}$ is an interval and 
$\ln f(x)$ is a concave function (on the interval $Q$).
\end{definition}

A theorem of Ibragimov \cite{Ibr56} gives some convolution properties of unimodal log-concave
functions. 

\begin{theorem} \label{thm:ibra} [Ibragimov] Let $f(x)$ be a unimodal log-concave pdf and $g(x)$ be
a unimodal pdf. The convolution $f * g$ is also unimodal.
\end{theorem}

\begin{claim} \label{clm:uni} The pdf $f_v(x)$ is unimodal.
\end{claim}

\begin{proof} We have $\ln \rho_v = \sum_i \ln \halpha_i$.
By \Clm{basic-xi}, the pdf of $\ln \halpha_i$ is $h_i(x) = e^x/2b_i$. 
Note that $h_i(x)$ is unimodal. Furthermore, $\ln h_i(x) = x - \ln 2b_i$,
which is concave. Since $\ln \rho_v$ is the sum of independent random variables, the pdf $f_v(x)$
is the convolution of the individual pdfs.
Repeated applications of Ibragimov's theorem (\Thm{ibra}) tells us that $f_v(x)$ is unimodal.
\end{proof}

\begin{proof} (of \Clm{min}) By the unimodality of $f_v$, $f_v$ is either non-decreasing, non-increasing,
or non-decreasing and then non-increasing in the interval $[x_1, x_2]$. Regardless of which case, for any $y \in [x_1, x_2]$,
$f(y) \geq \min(f(x_1),f(x_2))$.
\end{proof}

\subsection{Basic claims for NSKG} \label{sec:nskg-basic}

We now reprove some of the basic claims for NSKG. Note that
when we look at $\EX[X_d]$, the expectation is over both the randomness in $T$
and the edge insertions. We use $\bT$ to denote the set of matrices $T_1, T_2, \ldots, T_\ell$.
Conditioning on $\bT$ simply means conditioning on a fixed choice of the noise.

\begin{claim} \label{clm:nprob-pr} Let vertex $v \in S_r$. Choose the noise for NSKG
at random, and let $\hat p_v$ be the probability (conditioned on $\bT$) that a single edge insertion
produces an out-edge at $v$. (Note that $\hat p_v$ is itself a random variable, where the dependence
on $\bT$ is given by $\rho_v$.)
\begin{displaymath}
  \hat p_v = \frac{(1-4\sigma^2)^{\ell/2} \tau^r\rho_v}{n}.
\end{displaymath}
\end{claim}

\begin{myproof} This is identical to the proof of \Clm{prob-pr}. Consider a single edge insertion. 
For an edge insertion to be incident to $v$, the edge must go into the half
corresponding to the binary representation of $v$. If the $i$th bit of $v$ is $0$,
then the edge should drop in the top half at this level, and this happens with probability $(1/2+\sigma_i)$. 
On the other hand, if this bit is $1$, then the edge should drop in the bottom half,
which happens with probability $(1/2 - \sigma_i)$.
Let the bit representation of $v$ be $(z_1,z_2,\ldots,z_\ell)$. Then,
\begin{eqnarray*}
\hat p_v  & = & \prod_{i: z_i = 0} \left(\frac{1}{2} + \sigma_i \right) \prod_{i: z_i = 1} \left(\frac{1}{2} - \sigma_i \right)
= \prod_{i: z_i = 0} \alpha_i\left(\frac{1}{2} + \sigma \right) \prod_{i: z_i = 1} \beta_i\left(\frac{1}{2} - \sigma \right) \\
& = & \rho_v \left(\frac{1}{2} + \sigma\right)^{\ell/2+r}\left(\frac{1}{2} - \sigma\right)^{\ell/2-r} 
= \frac{\rho_v (1-4\sigma^2)^{\ell/2}}{2^\ell} \cdot \left(\frac{1/2+\sigma}{1/2-\sigma}\right)^r 
= \frac{(1-4\sigma^2)^{\ell/2} \tau^r \rho_v}{n}.
\end{eqnarray*}
\qquad\myproofend
\end{myproof}

As before, we will assume that $\hat p_v = o(1/\sqrt{m})$ and $d = o(\sqrt{n})$. Even though
$\hat p_v$ is a random variable, the probability that it is larger than $1/\sqrt{m}$ can be 
neglected. (This was discussed in more detail before \Lem{prob-d}). We stress that in the following,
the probability that $v$ has outdegree $d$ is itself a random variable. 

\begin{claim} \label{clm:nprob-d} Let $v$ be a vertex in slice $r$, $d = o(\sqrt{n})$, and $\hat p_v = o(1/\sqrt{m})$. 
Then for NSKG, we have
\begin{displaymath}
  \pr[\odeg{v} = d | \bT] = (1 \pm o(1))
   \frac{(\lambda_v)^d}{d!} \cdot \frac{(\tau^r)^d}{\exp(\lambda_v \tau^r)} 
\end{displaymath}
\end{claim}

\begin{myproof} We follow the proof of \Lem{prob-d}. 
We approximate ${m\choose d}$ by $m^d/d!$ and 
$(1-x)^{m-d}$ by $e^{-xm}$, for $x = o(1/\sqrt{m})$ and $d = o(\sqrt{n})$.
This approximation is performed in the first step below.
We remind the reader that $\lambda_v = \lambda \rho_v$.
By \Clm{nprob-pr} and the above approximations,
\begin{align*}
{m\choose d} \hat p_v^d (1-\hat p_v)^{m-d} 
& = {m\choose d} \left(\frac{(1-4\sigma^2)^{\ell/2} \tau^r\rho_v}{n}\right)^d \left(1 - \frac{(1-4\sigma^2)^{\ell/2} \tau^r\rho_v}{n}\right)^{m-d} \\
& = (1\pm o(1)) \frac{m^d}{d!} \cdot \left(\frac{(1-4\sigma^2)^{\ell/2} \tau^r\rho_v}{n}\right)^d \\
& \phantom{\geq} \cdot \exp\left(- \frac{(1-4\sigma^2)^{\ell/2} \tau^r\rho_vm}{n} \right) \\
& = (1 \pm o(1)) \frac{[\Delta (1-4\sigma^2)^{\ell/2}\rho_v]^d \tau^{rd}}{d!} \exp(- \Delta (1-4\sigma^2)^{\ell/2}\rho_v \tau^r) \\
& = (1 \pm o(1))  \frac{(\lambda\rho_v)^d}{d!} \cdot \frac{(\tau^r)^d}{\exp(\lambda \rho_v \tau^r)}. \qquad\myproofend
\end{align*}
\end{myproof}

\subsection{Bounds for degree distribution} \label{sec:formula-nskg}

We complete the proof of \Thm{lognormal}.
We break it down into some smaller claims. By and large,
the flow of the proof is similar to that for the standard SKG. The main
difference comes because the probabilities discussed in \Clm{nprob-d} are
random variables depending on the noise. The following claim is fairly
straightforward, given the previous analysis of standard SKG. This
is where we apply the Taylor approximations to show the Gaussian behavior
depicted in \Fig{gaussian}.

\begin{claim} \label{clm:sum} Consider some setting of the NSKG noise. Define $g_v(r) = r\ln \tau - \ln (d/\lambda_v)$. 
The expected number of vertices of degree $d$ conditioned on $\bT$ is
$$ \EX[X_d | \bT] =  \frac{1 \pm o(1)}{\sqrt{2\pi d}} \sum_{r = -\ell/2}^{\ell/2} \sum_{v \in S_r} \exp\left[- dg_v(r)^2/2\right] $$
\end{claim}

\begin{proof} By fixing some $\bT$, the $\lambda_v$s are fixed.
We use \Clm{nprob-d}, linearity of expectation, and Stirling's approximation in the following.
\begin{eqnarray*} \label{eq:xd-n} 
\EX_G[X_d] & = & \sum_{r = -\ell/2}^{\ell/2} \sum_{v \in S_r} (1 \pm o(1))\frac{{\lambda^d_v}}{d!} 
  \frac{(\tau^r)^d}{\exp(\lambda_v \tau^r)} \\
& = &  \frac{1 \pm o(1)}{\sqrt{2\pi d}} \sum_{r = -\ell/2}^{\ell/2} \sum_{v \in S_r} \left(\frac{e\lambda_v}{d}\right)^d 
\frac{(\tau^r)^d}{\exp(\lambda_v \tau^r)} 
\end{eqnarray*}
Choose a $v \in S_r$. 
\begin{displaymath} 
\left(\frac{e\lambda_v}{d}\right)^d  \frac{(\tau^r)^d}{\exp(\lambda_v \tau^r)} 
= \exp(d+d\ln \lambda_v + rd\ln \tau - d\ln d - \lambda_v \tau^r).
\end{displaymath}
Define $f_v(r) = rd\ln \tau - \lambda_v \tau^r - d\ln d + d\ln \lambda_v + d$,
where $r$ is an integer. 
We have $r = (\ln d - \ln \lambda_v + g_v(r))/\ln \tau$. 
\begin{eqnarray*}
f_v(r) & = & d\ln d - d\ln \lambda_v + dg_v(r) - e^{g_v(r)} d - d\ln d + d\ln \lambda_v + d \\
& = & d(1+g_v(r) - e^{g_v(r)}). 
\end{eqnarray*}
If $|g_v(r)| < 1$, then we can approximate $f_v(r) = -d[g_v(r)^2/2 + \Theta(g_v(r)^3)]$, and
get $\exp(f_v(r)) = (1 \pm o(1)) \exp(-dg_v(r)^2/2)$. This is analogous to the beginning of the proof of \Clm{rest}.
Suppose $|g_v(r)| \geq 1$. Then, arguing as in the proof of \Clm{most}, we deduce that $\exp(f_v(r)) \leq 2^{-\ell}$.
The sum of all these terms over $v$ is just a lower order term. So, we can substitute this by $\exp(-dg_v(r)^2/2)$.
Hence, we can bound
$$ \EX[X_d | \bT] =  \frac{1 \pm o(1)}{\sqrt{2\pi d}} \sum_{r = -\ell/2}^{\ell/2} \sum_{v \in S_r} \exp\left[- dg_v(r)^2/2\right] \qquad \myproofend $$
\end{proof}

We now reach the main challenge of this proof. The quantity $\EX[\exp(- dg_v(r)^2/2)]$
is evaluated by averaging over all noise. 
Note that the actual graph has no effect on this quantity.

\begin{lemma} \label{lem:norm} Consider $r = \Gamma_d = \nint{\theta_d}$.  
$$ \EX[\exp(- dg_v(r)^2/2)] 
\geq \frac{\nufour(c)}{\sqrt{d}} $$
\end{lemma}

\begin{proof} Define $\xi_{r,d} = (r-\theta_d)\ln \tau$.
Since $\theta_d = \ln(d/\lambda)/\ln \tau$,
$$ g_v(r) = r\ln \tau - \ln(d/\lambda_v) = r\ln \tau - \ln(d/\lambda) + \ln \rho_v
= \xi_{r,d} + \ln \rho_v $$
Hence,
$$ \EX[\exp(- dg_v(r)^2/2)] = \EX[\exp[-d(\ln \rho_v + \xi_{r,d})^2/2]] $$
Since we set $r = \nint{{\theta_d}}$, $|\xi_{r,d}| \leq (\ln \tau)/2$. 
Let us now evaluate the expectation. The pdf of $\ln \rho_v$
is denoted by $f_v$. The expectation is given by an integral.
To distinguish the $d$ referring to degree, and the $d$ referring to the infinitesimal, we shall use
$(d	)$ in parenthesis for the infinitesimal. We hope this slight abuse of notation will not create a problem, since our
integrals are not too confusing. 
By \Lem{rho-v}, $f_v(x) \geq \nuone(c)$ for $|x| \leq \htau$.
\begin{eqnarray*}
\EX[\exp(- dg_v(r)^2/2)] & = & \int^{+\infty}_{-\infty} \exp[-d(x+\xi_{r,d})^2/2] f_v(x) (dx) \\
& \geq & \nuone(c) \int^{\widehat{\tau}}_{-\htau} \exp[-d(x+\xi_{r,d})^2/2] (dx) \\
& = & \nuone(c) \int^{\htau + \xi_{r,d}}_{-\htau + \xi_{r,d}} \exp[-dx^2/2] (dx) \\
& = & \nuone(c) \Big[ \int^{+\infty}_{-\infty} \exp[-dx^2/2] (dx) - \int^{+\infty}_{\htau + \xi_{r,d}} \exp[-dx^2/2] (dx) \\
& & - \int^{-\htau + \xi_{r,d}}_{-\infty} \exp[-dx^2/2] \Big] (dx) 
\end{eqnarray*}
We have $|\xi_{r,d}| \leq (\ln \tau)/2$ and $\htau = \max(2,\ln \tau)$.
Hence, $\htau + \xi_{r,d} \geq 1$ and $-\htau + \xi_{r,d} \leq -1$.	
\begin{eqnarray*} 
\EX[\exp(- dg_v(r)^2/2)] & \geq & \nuone(c) \Big[ \int^{+\infty}_{-\infty} \exp[-dx^2/2] (dx) - \int^{+\infty}_{1} \exp[-dx^2/2] (dx) \\
& & - \int^{-1}_{-\infty} \exp[-dx^2/2] \Big] (dx) \\
& = & (\nuone(c)/\sqrt{d}) \Big[ \int^{+\infty}_{-\infty} e^{-x^2/2}dx - 2\int^{+\infty}_{\sqrt{d}} e^{-x^2/2}dx \Big]
\end{eqnarray*}
The first integral is just $\sqrt{2\pi}$. The second is a tail probability of the standard Gaussian,
bounded by $\int^{+\infty}_{y} e^{-x^2/2}dx < e^{-y^2/2}/y$ (Lemma 2, pg. 175 of~\cite{Fel50}). The second
term is at most $2e^{-d^2/2}/\sqrt{d} < \sqrt{\pi}$ (for sufficiently large $d$). Therefore,
we can set function $\nufour(c)$ such that $\EX[\exp(- dg_v(r)^2/2)] \geq \nufour(c)/\sqrt{d}$.
\end{proof}

\begin{proof}[of \Thm{lognormal}] This is a direct consequence of the previous claims.
Set $r = \Gamma_d$.
By \Clm{sum} and linearity of expectation, $\EX[X_d] = \EX[\EX[X_d | \bT]] \geq ((1-o(1))/\sqrt{2\pi d}) \sum_{v \in S_r} \EX[\exp(-d g_v(r)^2/2)]$.
\Lem{norm} tells us that $\EX[\exp(-d g_v(r)^2/2)] \geq \nufour(c)/\sqrt{d}$. Hence, $\EX[X_d] \geq \frac{\nu(c)}{d} {\ell \choose {\ell/2 + \Gamma_d}}$.
\end{proof}

\subsection{Subtleties in adding noise}
One might ask why we add noise in this particular fashion, and whether other ways
of adding noise are equally effective. Since we only need $\ell$ random numbers,
it seems intuitive that adding ``more noise" could only help. For example,
we might add noise on a per edge basis, i.e., at each level $i$ of \emph{every} edge insertion, we choose a new random perturbation $T_i$ of $T$. Interestingly, this
version of noise does not smooth out the degree distribution, as shown in \Fig{badnoise}.  In this figure, the red curve corresponds to the degree distribution of the graph generated by NSKG  with Graph500 parameters, $\ell=26$, and $b=0.1$.  The blue curve corresponds to generation  by adding noise per edge. 
As seen in this figure, adding noise per edge has hardly any effect on the
oscillations, while NSKG provides a smooth degree distribution curve. (These results are fairly consistent over different parameter choices.) It is crucial that we use the \emph{same} noisy $T_1,\ldots,T_\ell$
for every edge insertion.

\begin{figure}
  \centering
  \includegraphics[width=.5\columnwidth,trim=0 0 0 0,clip]{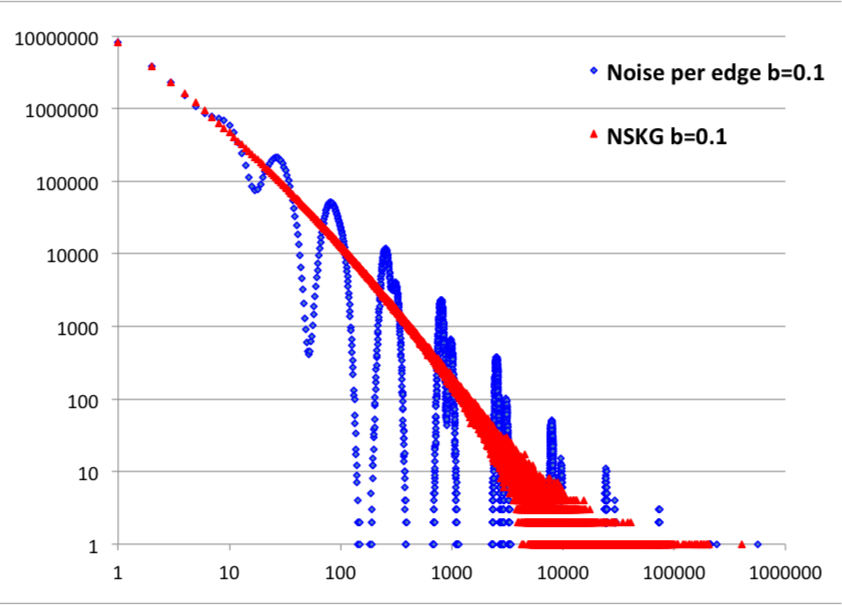}
  \caption{ Comparison of degree distribution of graphs generated by  NSKG  and by adding noise per edge for Graph500 parameters and $\ell=26$.  }
  \label{fig:badnoise}
\end{figure}

\section{Expected Number of Isolated Vertices}
\label{sec:vertices}

In this section, we give a simple formula for the number of isolated vertices in SKG.
We focus on the symmetric case,
where $t_2 = t_3$  in the  matrix $T$. We assume that $\ell$ is even
in the following, but the formula can be  extended  for 
$\ell$ being odd.
The real  contribution here is 
a clearer understanding of  how  many vertices  SKG leaves  isolated and how the SKG parameters affects this number.

\begin{theorem} \label{thm:est-isol} Consider SKG with $T$ symmetric and let $I$ denote the
  number of isolated vertices. With probability $1-o(1)$,
  \begin{equation}
    \label{eq:isovert}
    I = (1 \pm o(1))\sum_{r = -\ell/2}^{r = \ell/2} {\ell\choose \ell/2+r} \exp(-2\lambda\tau^r) .
  \end{equation}
\end{theorem}

\begin{claim} \label{clm:prob-inc-edge} %
  Let $q_r$ be the probability that a single edge insertion produces an
  in-edge or out-edge incident to $v \in S_r$. Then, for SKG with $T$ symmetric, 
  \begin{displaymath}
    q_r = (1 \pm o(1))\frac{2(1-4\sigma^2)^{\ell/2} \tau^r}{n}.
  \end{displaymath}
\end{claim}

\begin{proof} Let $\cE_o$ (resp. $\cE_i$) be the event that a single edge insertion is an in-edge (resp. out-edge)
of $v$. We have $q_r = \Pr(\cE_o) + \Pr(\cE_i) - \Pr(\cE_o \cup \cE_i)$.
By \Clm{prob-pr} and the symmetry to $T$, the first two probabilities are $\frac{(1-4\sigma^2)^{\ell/2} \tau^r}{n}$.
The last is the probability that the edge insertion leads to a self-loop at $v$.
This is at most $\sigma^\ell \Pr(\cE_o)$. Since $\sigma < 1$, this is $o(\Pr(\cE_o))$.
\end{proof}

As before, we can assume that $q_r \leq 1/\sqrt{m}$. By \Clm{large-d}, if $q_r \geq p_r \geq 1/\sqrt{m}$,
then with probability tending to $1$, vertices in slice $r$ are not isolated. Hence, we
can ignore such vertices when computing estimates for $I$.

\begin{claim} \label{clm:prob-isol} %
  Let $v \in S_r$ and assume $q_r \leq 1/\sqrt{m}$. Then, for SKG with
  $T$ symmetric,
  \begin{displaymath}
    \pr[\text{\rm $v$ is isolated}] =
    (1 \pm o(1))\exp(-2\lambda\tau^r).
  \end{displaymath}
\end{claim}

\begin{myproof} %
  Using \Clm{prob-inc-edge} and $(1 - x)^m = (1 \pm o(1))e^{-xm}$, for $|x| \leq 1/\sqrt{m}$,
  \begin{displaymath}
    (1-q_r)^m = (1 \pm o(1))\exp(-2(1 \pm o(1))\Delta(1-4\sigma^2)^{\ell/2} \tau^r) = 
    (1 \pm o(1))\exp(-2(1 \pm o(1))\lambda \tau^r).
  \end{displaymath}
  For large $\ell$, this converges to $\exp(-2\lambda\tau^r)$.
\quad\myproofend
\end{myproof}

\begin{proof}[of \Thm{est-isol}] %
By \Clm{prob-isol}
and linearity of expectation, the expected number of isolated vertices
is 
\begin{displaymath}
  (1 \pm o(1)) \sum_{r = -\ell/2}^{r = \ell/2} {\ell\choose \ell/2+r} \exp(-2\lambda\tau^r).
\end{displaymath}
To bound that actual number of isolated vertices, we use concentration inequalities
for functions of independent random variables. Let $Y$ denote the number of
isolated vertices, and $X_1, X_2, \ldots, X_m$ be the labels of the $m$
edge insertions. Note that all the $X_i$'s are independent, and $Y$
is some fixed function of $X_1, X_2, \ldots, X_m$. Suppose we fix all the
edge insertions and just modify one insertion. Then, the number of
isolated vertices can change by at most $c = 2$. Hence, the function
defining $Y$ satisfies a Lipschitz condition. 
This means that changing a single argument of $Y$ (some $X_i$) modifies
the value of $Y$ by at most a constant ($c$).
By McDiarmid's inequality~\cite{Mc89},
\begin{displaymath}
  \pr[|Y - \EX[Y]| > \epsilon] < 2\exp\left(-\frac{2\epsilon^2}{c^2m}\right).
\end{displaymath}
Setting $\epsilon = \sqrt{m \log m}$, we get the probability that $Y$ deviates
from its expectation by more than $\sqrt{m\log m}$ is $o(1)$. 
The expected number of vertices is at least ${\ell \choose \ell/2} \exp(-2\lambda)$,
and $\sqrt{m \log m}$ is a lower order term with respect to this quantity. This
completes the proof.
\end{proof}

The fraction of isolated vertices in a slice $r$ is essentially
$\exp(-\lambda \tau^r)$. Note that $\tau$ is larger than $1$.
Hence, this is a decreasing function of $r$. This is quite natural,
since if a vertex $v$ has many zeros in its representation (higher slice),
then it is likely to have a larger degree (and less likely to be isolated).
This function is doubly exponential in $r$, and therefore decreases  quickly with $r$.
The fraction of isolates rapidly goes to $0$ (resp. $1$) as $r$ is 
positive (resp. negative).

\subsection{Effect of noise on isolated vertices}

The introduction of noise was quite successful in correcting the degree distribution but has little effect
on the number of isolated vertices. This is not surprising, considering the noise affects fat tail behavior
of the degree distribution. The number of isolated vertices is a  different aspect of the degree distribution.
The data  presented in \Tab{noise:iso}  clearly shows that the number of isolated vertices is quite resistant to noise.  While there is some decrease in the  number of isolated vertices,  this quantity is very small compared to the total number of isolated vertices.  We have observed similar results on the other parameter settings. 
\begin{table}[h]
\centering 
\caption{\label{tab:noise:iso}  Percentage of isolated vertices with different noise levels for the GRAPH500 parameters and $\ell=26$} 
\begin{tabular}{c|c}
\hline 
Max. noise level ($b$) &  \% isolated vertices \\\hline 
0 & 	51.12 \\
0.05	& 49.26 \\
0.06	& 49.12 \\
0.07	& 49.06 \\ 
0.08	& 49.07 \\
0.09	& 49.16 \\
0.1	& 49.34 \\\hline
\end{tabular}
\end{table}

In addition to this empirical study, we can also give some mathematical intuition behind these observations. The equivalent statement of \Clm{prob-isol}
for NSKG is 
$$ \pr[\text{\rm $v$ is isolated}] \geq (1-o(1))\exp(-2\lambda\tau^r) = (1-o(1))[\exp(-2\lambda\tau^r)]^{\rho_v} $$
The noiseless version of this probability is $[\exp(-2\lambda\tau_r)]$.
Note that the probability now is a random variable that depends on $T$, since $\rho_v$ depends
on the noise. \Lem{norm} tells us that $\ln \rho_v$ lies mostly in the range $[1-c'/\sqrt{\ell},1+c'/\sqrt{\ell}]$
(for constant $c'$), and is concentrated close to $1$.

We are mainly interested in the case when the probability that $v$ is isolated
is \emph{not} vanishingly small (is at least, say $0.01$). As $\ell$ grows, $\rho_v$ is close
to being $1$, and deviations are quite small. So, when we take the noiseless probability to the $\rho_v$th
power, we get almost the same value.

\subsection{Relation of SKG parameters to the number of isolated vertices:} 
When $\lambda$ decreases, the number
of isolated vertices increases. Suppose we fix the SKG matrix and average
degree $\Delta$, and start increasing $\ell$. Note that this is done
in the Graph500 benchmark, to construct larger and larger graphs. The value
of $\lambda$ decreases exponentially in $\ell$, so the number
of isolated vertices will increase. 
Our formula suggests ways of counteracting
this problem. The value of $\Delta$ could be increased, or the value $\sigma$
could be decreased. But, in general, this will be a problem for generating
large sparse graphs using a fixed SKG matrix.

When $\sigma$ increases, then $\lambda$ decreases and $\tau$ increases.
Nonetheless, the effect of $\lambda$ is much stronger than that of $\tau$.
Hence, the number of isolated vertices will increase as $\sigma$ increases.
In \Tab{isol-500}, we compute the estimated number of isolated vertices
in graphs for  the Graph500 parameters. Observe how the
fraction of isolated vertices consistently increases as $\ell$ is
increased. For the largest setting of $k=42$, only one fourth
of the vertices are not isolated.

\section{$k$-cores in SKG}
\label{sec:cores}

Structures of $k$-cores are an important part 
of social network analysis \cite{CaHaKi+07,AlDaBa+08,KuNoTo10}, as they are a manifestation of the
community structure and high connectivity of these graphs.

\begin{definition} \label{def:core} Given an undirected graph $G = (V,E)$, 
the \emph{subgraph induced by set $S \subseteq V$},
is denoted by $G|_S := (S,E')$, where $E'$ contains every edge of $E$ that is completely
contained in $S$. 
For an undirected graph, the \emph{$k$-core} of $G$ the largest induced subgraph of minimum degree $k$. 
The \emph{max core number} of $G$ is the largest $k$ such that $G$ contains a (non-empty)
$k$-core. (These can be extended to directed versions: a $k$-out-core is a subgraph with
min out-degree $k$.)
\end{definition}

A bipartite core is an induced subgraph with every vertex
has \emph{either} a high in-degree or out-degree. The former
are called \emph{authorities} and the latter are \emph{hubs}.
Large bipartite cores are present in web graphs and
are an important structural component \cite{GiKlRa98,Kl99}.
Note that if we make the directed graph undirected (by simply
removing the directions), then a bipartite
core becomes a normal core. Hence, it is useful to compute cores 
in a directed graph by making it undirected.

We begin by comparing the sizes of $k$-cores in real graphs, and their models
using SKG \cite{LeChKlFa10}. Refer to \Fig{core-examples}. We plot the size
of the maximum $k$-core with $k$. The $k$ at which the curve ends is the
max core number. (For CAHepPh, we look at undirected cores, since this is
an undirected graph. For WEBNotreDame, a directed graph, we look at
out-cores. But the empirical observations we make holds for all
other core versions.)
For both our
examples, we see how drastically different the curves are.
By far the most important difference is that the curve for the SKG versions are
extremely short.
This means that the max core number is \emph{much smaller} for SKG modeled graphs
compared to their real counterparts. 
For the web graph WEBNotreDame,
we see the presence of large cores, probably an indication of some community structure.
The maximum core number of the SKG version is an \emph{order of magnitude} smaller.
Minor modifications (like increasing degree,
or slight variation of parameters) to these graphs do not increase
the core sizes or max cores numbers much. This is a problem,
since this is strongly suggesting that SKG do not exhibit localized
density like real web graphs or social networks.

\begin{figure}
  \centering
  \includegraphics[width=.45\columnwidth,trim=0 0 0 0,clip]{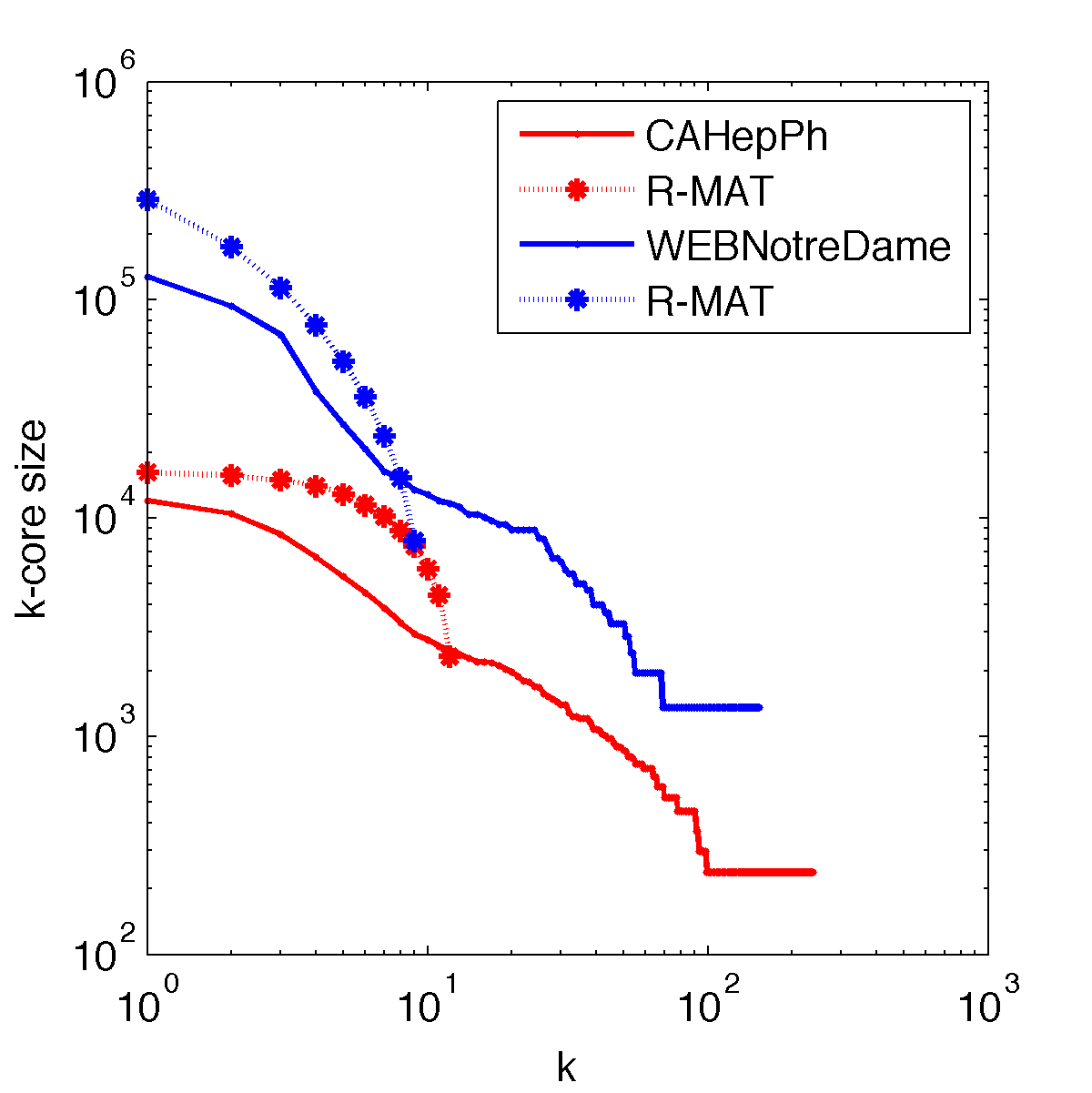}
  \caption{Core decompositions of real graphs and their SKG model. Observe that
  the max core  of SKG is an order of magnitude smaller. }
  \label{fig:core-examples}
\end{figure}

If we wish to use SKG to model real
networks, then it is imperative to understand the behavior
of max core numbers for SKG. Indeed, in \Tab{core}, we see that our observation
is not just an artifact of our examples. SKG consistently have very 
low max core number. Only for the peer-to-peer
Gnutella graphs does SKG match the real data, and this is specifically
for the case where the max core number is extremely small.
For the undirected graph (the first three co-authorship networks),
we have computed the undirected cores. The corresponding
SKG is generated by copying
the upper triangular part in the lower half to get a symmetric matrix (an undirected
graph). The remaining graphs are directed, and we simply remove the
direction on the edges and compute the total core.
Our observations hold for in and out cores as well\confer{(given in full version)}, and for  a wide range of data.
This is an indication
that SKG is not generating sufficiently dense subgraphs. 
\begin{table}[h] 
\caption{Core sizes in real graphs and SKG version} 
\centering %
\begin{tabular}{l c c} %
\hline 
Graph & Real max core & SKG max core\\ [0.5ex]
\hline %
CAGrQc & 43 & 4\\
CAHepPh & 238 & 16\\
CAHepTh & 31 & 5\\
CITHepPh & 30 & 19\\
CITHepTh & 37 & 19\\
P2PGnutella25 & 5 & 5\\
P2PGnutella30 & 7 & 6\\
SOCEpinions & 67 & 43\\
WEBNotreDame & 155 & 31\\%
\hline %
\end{tabular}
\label{tab:core}
\end{table}

We focus our attention on the max core number of SKG. How does this number
change with the various parameters?
The following summarizes our observations.

\begin{empirical} \label{emp:core} For SKG with symmetric $T$, we have the following observations.
\begin{asparaenum}
	\item The max core number increases with $\sigma$. By and large, if $\sigma < 0.1$, 
	max core numbers are extremely tiny.
	\item Max core numbers grow with $\ell$ only when the values of $\sigma$ are sufficiently large.
	Even then, the growth is much slower than the size of the graph. For smaller $\sigma$,
	max core numbers exhibit essentially negligible growth.
	\item Max core numbers increase essentially linearly with $\Delta$.
\end{asparaenum}

\vspace{2ex}
 Large max core numbers require larger values of $\sigma$. 
As mentioned in \Sec{vertices}, increasing
$\sigma$ increases the number of isolated vertices. 
Hence, there is an inherent tension between increasing the max core number and decreasing 
the number of isolated vertices. 
\end{empirical}

For the sake of consistency, we performed the following experiments on the max core
after taking a symmetric version of the SKG graph. 
Our results look the same for in and out cores as well.
In \Fig{core-beta}, we show how increasing $\sigma$ increases the max core number.
We fix the values of $\ell = 16$ and $m = 6 \times 2^{16}$.
(There is nothing special about these values. Indeed the results are basically
identical, regardless of this choice.) Then, we fix $t_1$ (or $t_2$) to some value, and slowly
increase $\sigma$ by increasing $t_2$ (resp. $t_1$). We see that regardless of the 
fixed values of $t_1$ (or $t_2$), the max core consistently increases. But as long
as $\sigma < 0.1$, max core numbers remain almost the same.

\begin{figure*}[t]
  \centering
  \subfloat[Varying $\sigma$]{\label{fig:core-beta}
  \includegraphics[width=.45\textwidth]{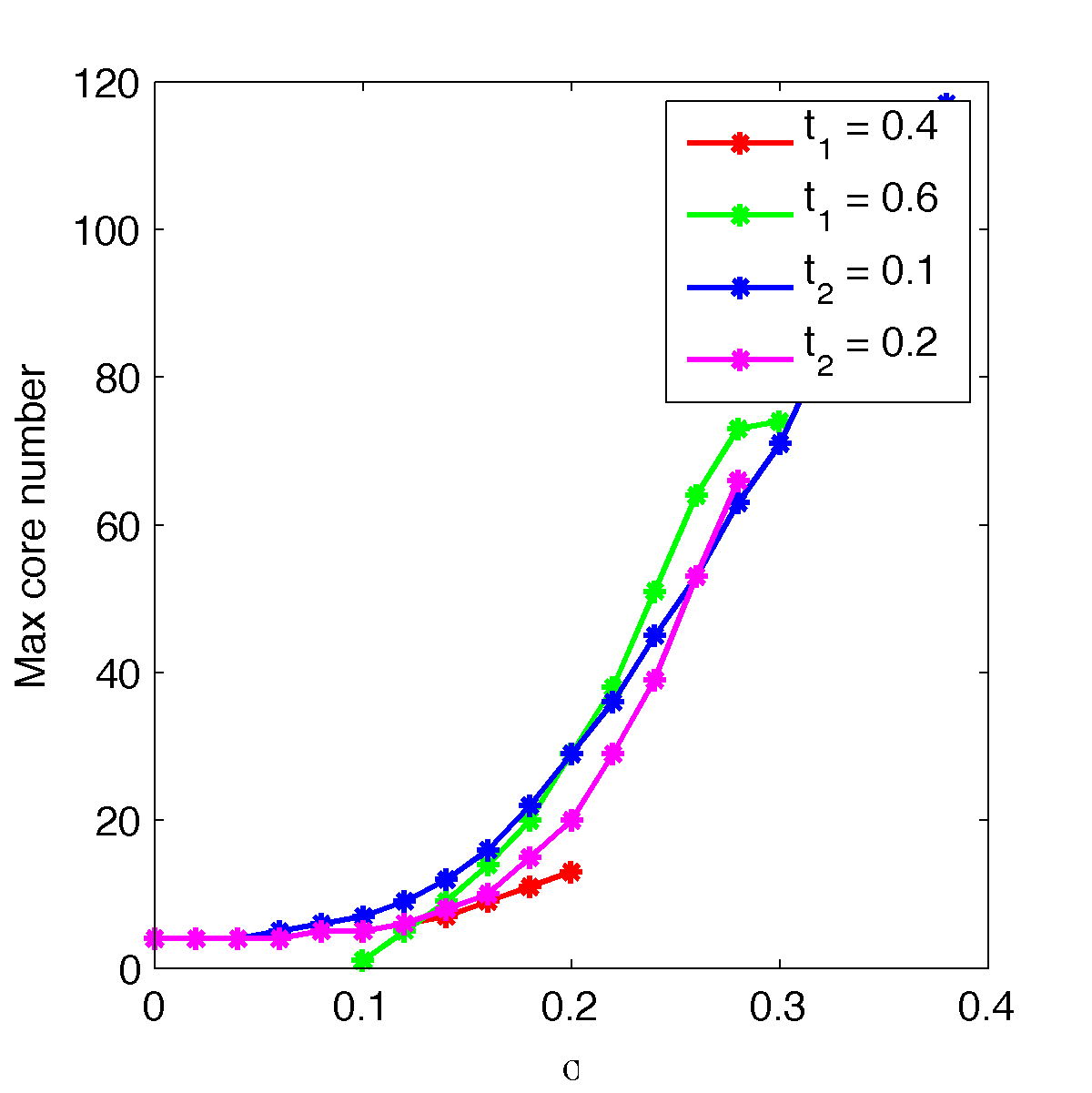}
  }
  \subfloat[Varying $\ell$]{\label{fig:core-k}
  \includegraphics[width=.45\textwidth]{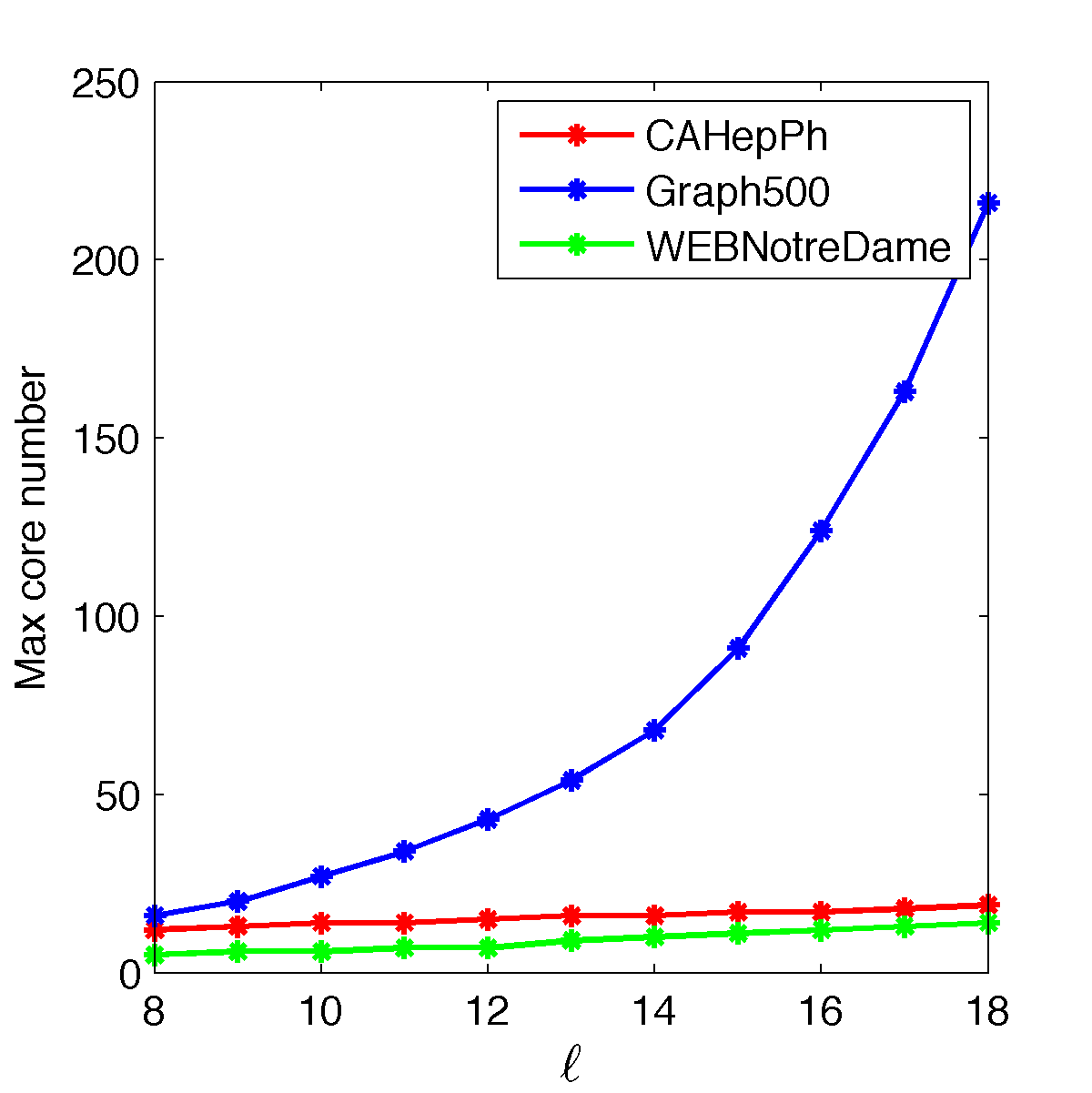}
  }

  \subfloat[Varying $\Delta$]{\label{fig:core-avg}
  \includegraphics[width=.45\textwidth]{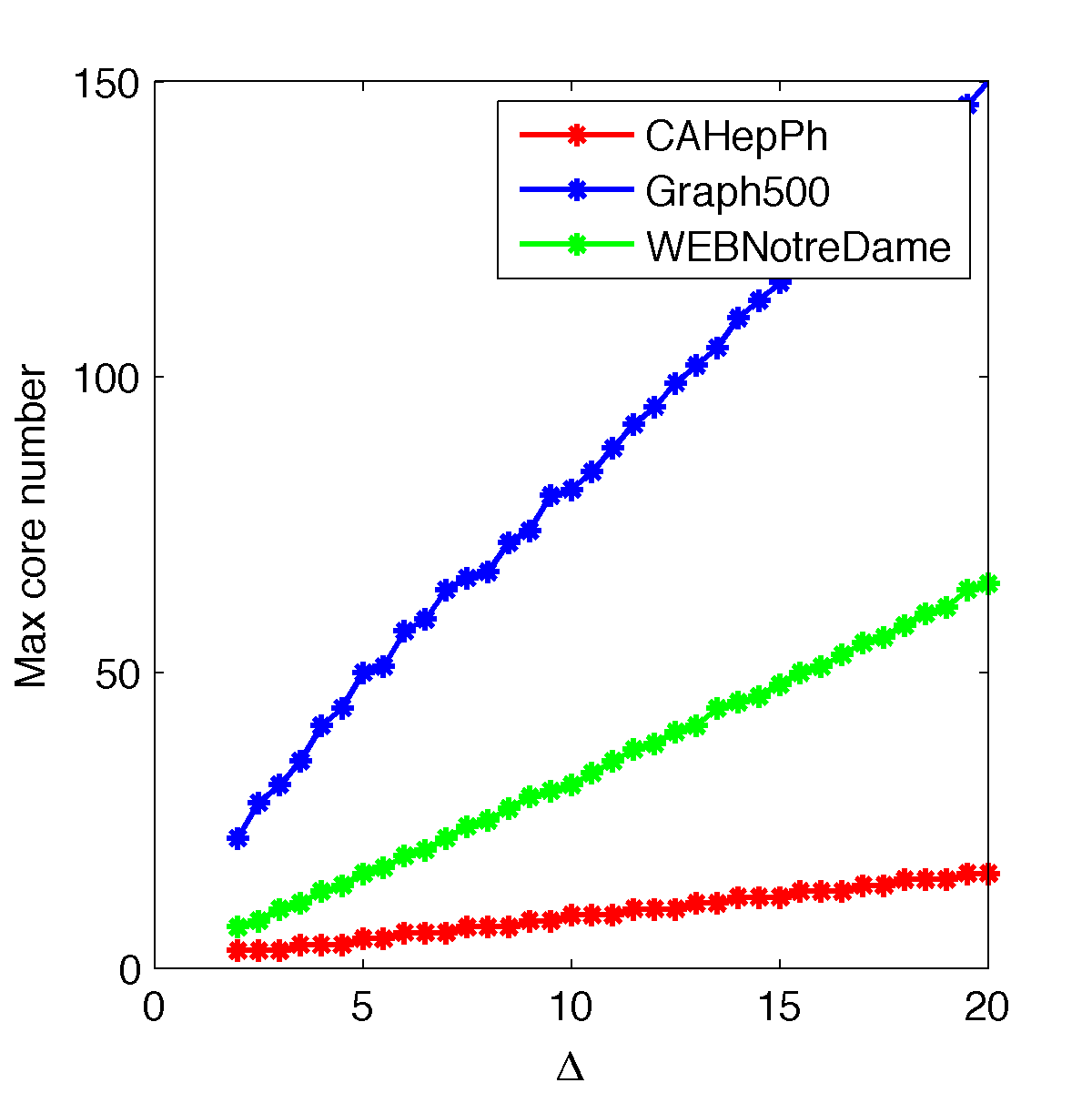}
  }
  \caption{We plot the max core number against various parameters. In the first picture, we plot
  the max core number of an (symmetric) SKG graph with increasing $\sigma$. Next, we show
  how the max core number increases with $\ell$, the number of levels. Observe the major role that the matrix
  $\sigma$ plays. For Graph500, $\sigma$ is much larger than the other parameter sets.
  Finally, we show that regardless of the parameters, the max core number increases linearly
  with $\Delta$.}
  \label{fig:core}
\end{figure*}

In \Fig{core-k}, we fix matrix $T$ and average degree $\Delta$, and only 
vary $\ell$. 
For WEBNotreDame\footnote{Even though the matrix $T$
is not symmetric, we can still define $\sigma$. Also, the off diagonal values
are $0.20$ and $0.21$, so they are almost equal.}, we have $\sigma = 0.18$
and for CA-HEP-Ph, we have $\sigma = 0.11$.
For both cases, increasing $\ell$ barely increases the max core
number. Despite increasing the graph  size  by $8$ orders of magnitude,
the max core number only doubles. Contrast this with the Graph500 setting, where 
 $\sigma = 0.26$, and we see a steady increase with larger $\ell$. This
is a predictable pattern we notice for many different parameter settings:
larger $\sigma$ leads to larger max core numbers as $\ell$ goes up.
Finally, in \Fig{core-avg}, we see that the max core number is basically
linear in $\Delta$.

\subsection{Effect of noise on cores} \label{sec:noise-core}

Our general intuition is that NSKG mainly redistributes
edges of SKG to get a smooth degree distribution, but does not have major effects
on the overall structure of the graph. This is somewhat validated by our studies on
isolated vertices and reinforced by looking at $k$-cores. In \Fig{noisycore}, we plot
the core decompositions of SKG and two versions on NSKG ($b = 0.05$ and $b=0.1$). 
We observe that there are little changes in these decompositions, although there
is a smoothening of the curve for Graph500 parameters. The problem of tiny cores of SKG
is not mitigated by the addition of noise.

 \begin{figure*}[t]
  \centering
  \subfloat[GRAPH500]{
  \includegraphics[width=.45\textwidth]{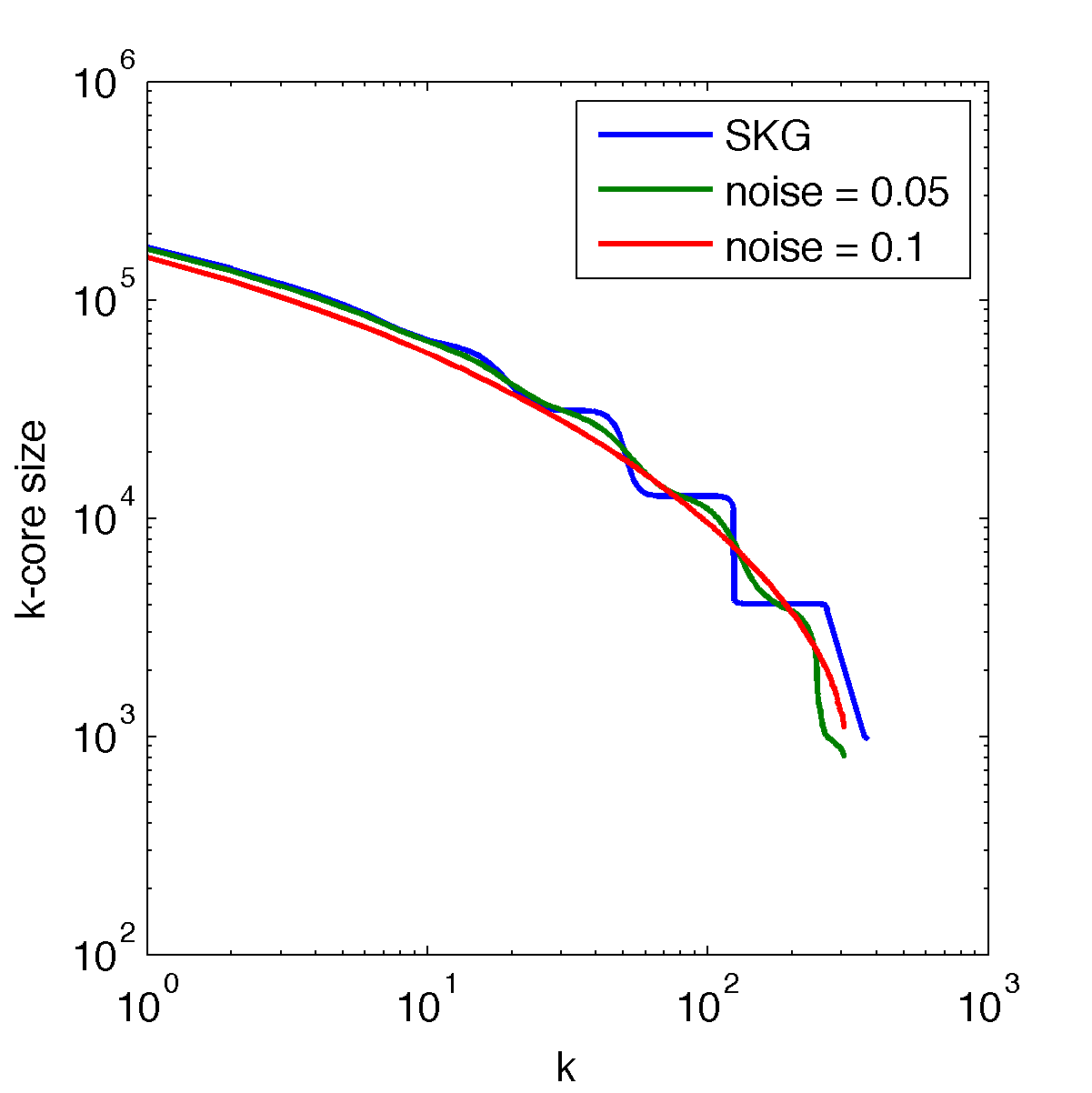}
  }
  \subfloat[WebNotreDame]{
  \includegraphics[width=.45\textwidth]{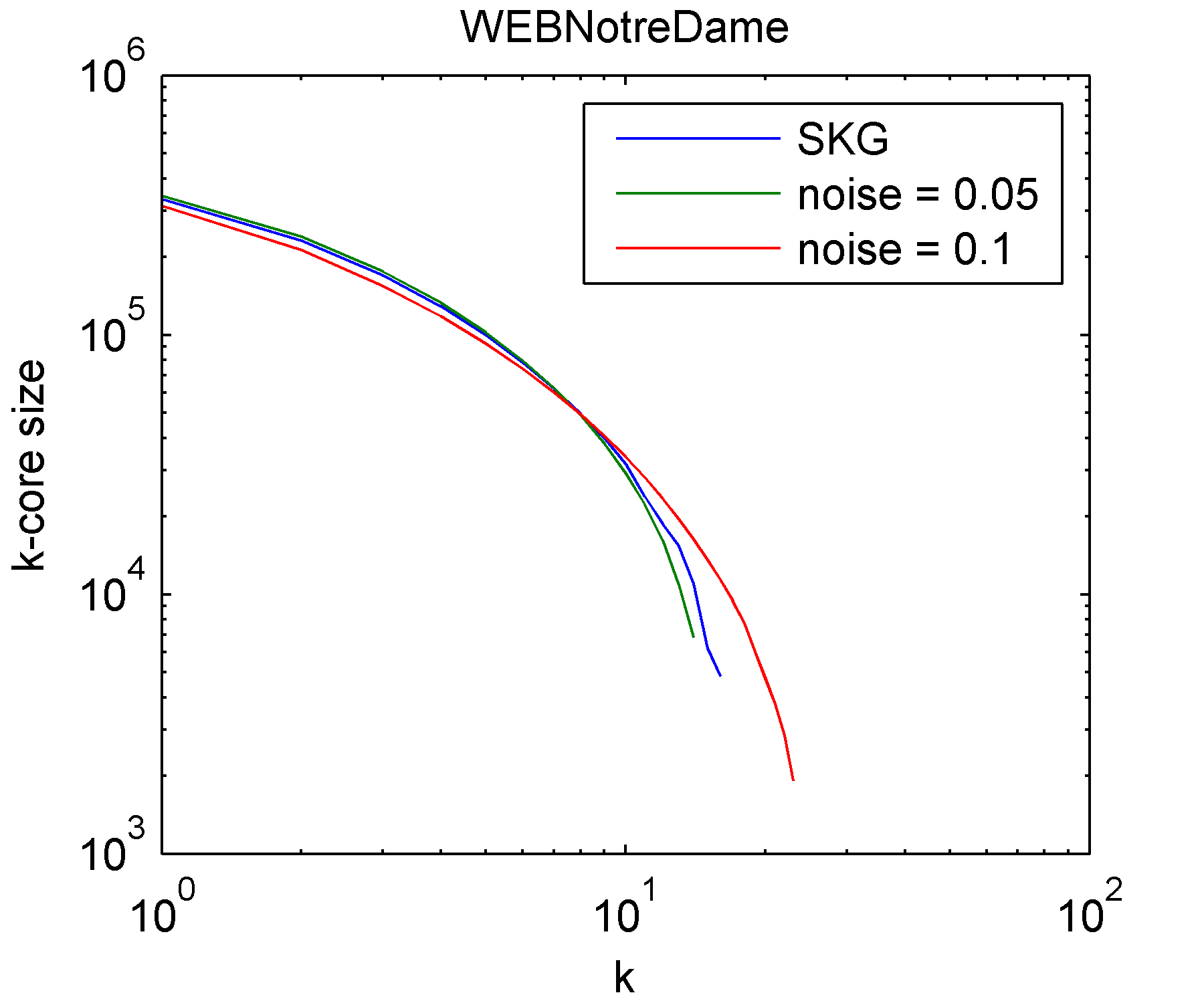}
  }
  
  \subfloat[CAHepPH]{
  \includegraphics[width=.45\textwidth]{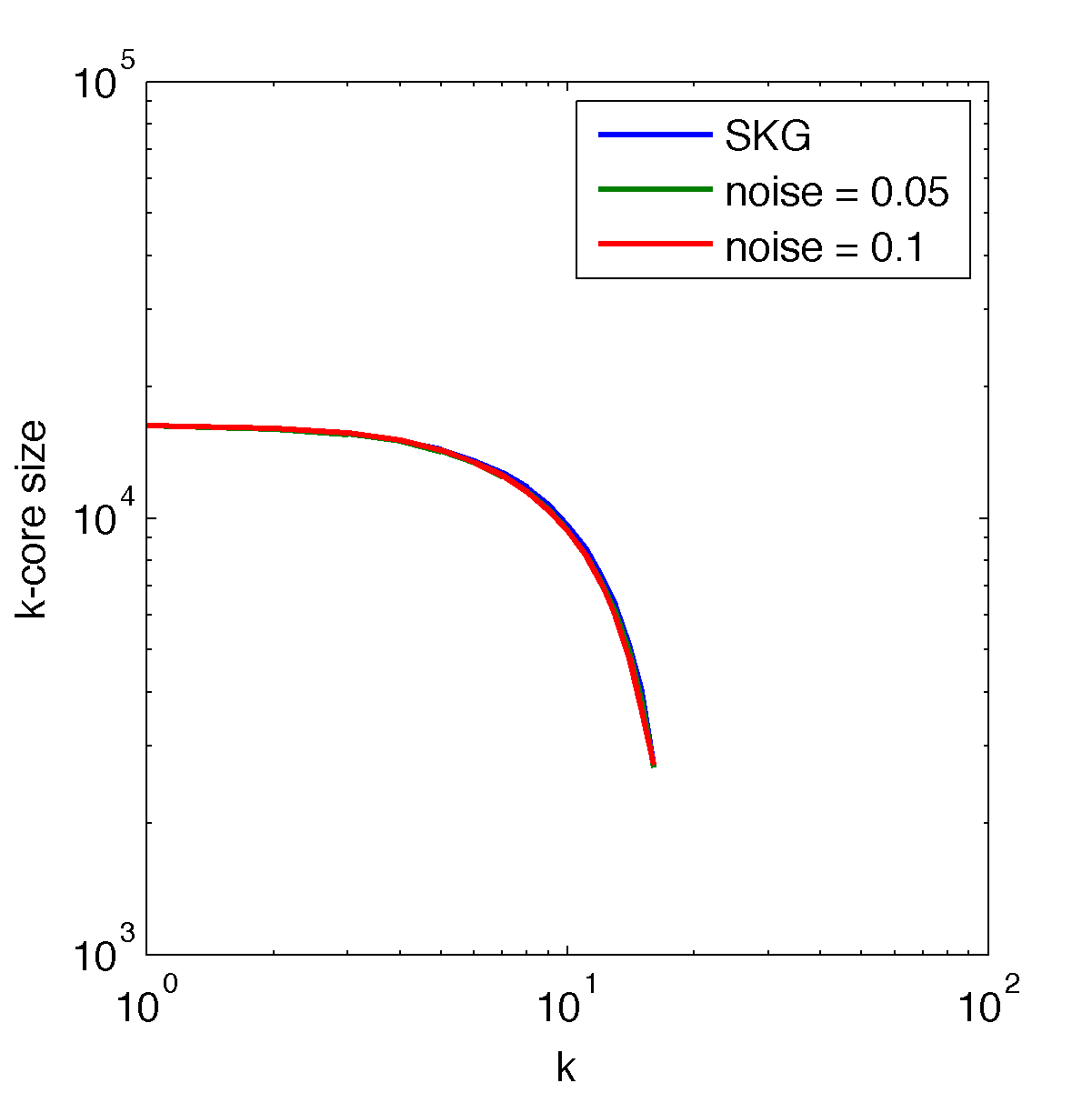}
  }
  \caption{We plot the core decomposition of SKG and NSKG (with 2 settings of noise) for the different
  parameters. Observe that there is only a minor change in core sizes with noise.}
  \label{fig:noisycore}
\end{figure*}

\section{Conclusions}
\label{sec:conclusions}
For a true understanding of a model, a careful theoretical and empirical study of
its properties  in relation to its parameters is imperative. This not
only provides insight into why certain properties arise, but also suggests
ways for enhancement.  One strength of the SKG model  is its amenability to rigorous analysis, which we exploit in this paper.  

We prove strong theorems about the degree distribution, and more significantly
show how adding noise can give a true lognormal distribution by eliminating the oscillations in degree distributions.
Our proposed method of adding noise requires only $\ell$ random numbers all together, and is hence cost effective. 
We want to stress that our major contribution is in providing \emph{both} the theory and
matching empirical evidence.  
The formula for expected number of isolated vertices provides an efficient alternative to methods  for computing 
the full  degree distribution.
Besides requiring fewer operations to compute and being less prone to numerical errors, 
the formula transparently relates the expected number of isolated vertices to  the SKG parameters.
Our studies on core numbers  establish a connection between the model parameters  and the cores of the  resulting graphs.  In particular, we show that  commonly used  SKG parameters  generate tiny  cores,  and the model's ability to generate large cores is  limited.

\begin{ack}
We are grateful to David Gleich for the MATLAB BGL library as well as many helpful discussions. 
We thank Todd Plantenga for creating large SKG and NSKG instances, and for generating \Fig{noisy_degdist_graph500}.
We also thank Jon Berry for checking our Graph500 predictions against real data, 
and also David Bader and Richard Murphy for discussions about the Graph500 benchmark.
We acknowledge the inspiration of Jennifer Neville and Blair Sullivan, who inspired us with their different work on SKG during recent visits to Sandia.
\end{ack}


\begin{thebibliography}{}

\bibitem[\protect\citeauthoryear{Alvarez-Hamelin, Dall'Asta, Barrat, and
  Vespignani}{Alvarez-Hamelin et~al\mbox{.}}{2008}]{AlDaBa+08}
{\sc Alvarez-Hamelin, J.~I.}, {\sc Dall'Asta, L.}, {\sc Barrat, A.}, {\sc and}
  {\sc Vespignani, A.} 2008.
\newblock K-core decomposition of internet graphs: hierarchies, self-similarity
  and measurement biases.
\newblock {\em Networks and Heterogenous Media\/}~{\em 3,\/}~2, 371--393.

\bibitem[\protect\citeauthoryear{Andersen and Chellapilla}{Andersen and
  Chellapilla}{2009}]{AnCh09}
{\sc Andersen, R.} {\sc and} {\sc Chellapilla, K.} 2009.
\newblock Finding dense subgraphs with size bounds.
\newblock In {\em Algorithms and Models for the Web-Graph}. Springer, 25--37.

\bibitem[\protect\citeauthoryear{Berry}{Berry}{1941}]{Ber41}
{\sc Berry, A.} 1941.
\newblock The accuracy of the gaussian approximation to the sum of independent
  variates.
\newblock {\em Transactions of the American Mathematical Society\/}~{\em
  49,\/}~1, 122--136.

\bibitem[\protect\citeauthoryear{Bi, Faloutsos, and Korn}{Bi
  et~al\mbox{.}}{2001}]{BiFaKo01}
{\sc Bi, Z.}, {\sc Faloutsos, C.}, {\sc and} {\sc Korn, F.} 2001.
\newblock The ``{DGX}" distribution for mining massive, skewed data.
\newblock In {\em KDD~'01}. ACM, 17--26.

\bibitem[\protect\citeauthoryear{Carmi, Havlin, Kirkpatrick, Shavitt, and
  Shir}{Carmi et~al\mbox{.}}{2007}]{CaHaKi+07}
{\sc Carmi, S.}, {\sc Havlin, S.}, {\sc Kirkpatrick, S.}, {\sc Shavitt, Y.},
  {\sc and} {\sc Shir, E.} 2007.
\newblock A model of internet topology using k-shell decomposition.
\newblock {\em PNAS\/}~{\em 104,\/}~27, 11150--11154.

\bibitem[\protect\citeauthoryear{Chakrabarti and Faloutsos}{Chakrabarti and
  Faloutsos}{2006}]{ChFa06}
{\sc Chakrabarti, D.} {\sc and} {\sc Faloutsos, C.} 2006.
\newblock Graph mining: Laws, generators, and algorithms.
\newblock {\em ACM Computing Surveys\/}~{\em 38,\/}~1.

\bibitem[\protect\citeauthoryear{Chakrabarti, Zhan, and Faloutsos}{Chakrabarti
  et~al\mbox{.}}{2004}]{ChZhFa04}
{\sc Chakrabarti, D.}, {\sc Zhan, Y.}, {\sc and} {\sc Faloutsos, C.} 2004.
\newblock {R-MAT}: A recursive model for graph mining.
\newblock In {\em SDM '04}. 442--446.

\bibitem[\protect\citeauthoryear{Clauset, Shalizi, and Newman}{Clauset
  et~al\mbox{.}}{2009}]{ClShNe09}
{\sc Clauset, A.}, {\sc Shalizi, C.~R.}, {\sc and} {\sc Newman, M. E.~J.} 2009.
\newblock Power-law distributions in empirical data.
\newblock {\em SIAM Review\/}~{\em 51,\/}~4, 661--703.

\bibitem[\protect\citeauthoryear{Esseen}{Esseen}{1942}]{Ess42}
{\sc Esseen, C.-G.} 1942.
\newblock A moment inequality with an application to the central limit theorem.
\newblock {\em Skand. Aktuarietidskr\/}~{\em 39}, 160--170.

\bibitem[\protect\citeauthoryear{Feller}{Feller}{1968}]{Fel50}
{\sc Feller, W.} 1968.
\newblock {\em An Introduction to probability theory and applications: Vol I\/}
  3rd Ed.
\newblock John Wiley and Sons.

\bibitem[\protect\citeauthoryear{Gibson, Kleinberg, and Raghavan}{Gibson
  et~al\mbox{.}}{1998}]{GiKlRa98}
{\sc Gibson, D.}, {\sc Kleinberg, J.}, {\sc and} {\sc Raghavan, P.} 1998.
\newblock Inferring web communities from link topology.
\newblock In {\em HYPERTEXT '98}. ACM, 225--234.

\bibitem[\protect\citeauthoryear{Gkantsidis, Mihail, and Zegura}{Gkantsidis
  et~al\mbox{.}}{2003}]{GkMi03}
{\sc Gkantsidis, C.}, {\sc Mihail, M.}, {\sc and} {\sc Zegura, E.~W.} 2003.
\newblock Spectral analysis of internet topologies.
\newblock In {\em INFOCOM 2003}. IEEE, 364--374.

\bibitem[\protect\citeauthoryear{Goltsev, Dorogovtsev, and Mendes}{Goltsev
  et~al\mbox{.}}{2006}]{GoDo06}
{\sc Goltsev, A.~V.}, {\sc Dorogovtsev, S.~N.}, {\sc and} {\sc Mendes, J.
  F.~F.} 2006.
\newblock $k$-core (bootstrap) percolation on complex networks: Critical
  phenomena and nonlocal effects.
\newblock {\em Phys. Rev. E\/}~{\em 73,\/}~5, 056101.

\bibitem[\protect\citeauthoryear{{Graph500 Steering Committee}}{{Graph500
  Steering Committee}}{2012}]{Graph500}
{\sc {Graph500 Steering Committee}}. 2012.
\newblock Graph 500 benchmark.
\newblock Available at \url{http://www.graph500.org/specifications}.

\bibitem[\protect\citeauthoryear{Gro\"{e}r, Sullivan, and Poole}{Gro\"{e}r
  et~al\mbox{.}}{2011}]{GrSuPo10}
{\sc Gro\"{e}r, C.}, {\sc Sullivan, B.~D.}, {\sc and} {\sc Poole, S.} 2011.
\newblock A mathematical analysis of the {R}-{MAT} random graph generator.
\newblock {\em Networks\/}~{\em 58,\/}~3, 159--170.

\bibitem[\protect\citeauthoryear{Ibragimov}{Ibragimov}{1956}]{Ibr56}
{\sc Ibragimov, I.~A.} 1956.
\newblock On the composition of unimodal distributions.
\newblock {\em Theory of Probability and its Applications / Volume 1 / Issue
  2\/}~{\em 1,\/}~2, 255--260.

\bibitem[\protect\citeauthoryear{Kim and Leskovec}{Kim and
  Leskovec}{2010}]{KiLe10}
{\sc Kim, M.} {\sc and} {\sc Leskovec, J.} 2010.
\newblock Multiplicative attribute graph model of real-world networks.
\newblock arXiv:1009.3499v2.

\bibitem[\protect\citeauthoryear{Kleinberg}{Kleinberg}{1999}]{Kl99}
{\sc Kleinberg, J.~M.} 1999.
\newblock Authoritative sources in a hyperlinked environment.
\newblock {\em J. ACM\/}~{\em 46,\/}~5, 604--632.

\bibitem[\protect\citeauthoryear{Kumar, Novak, and Tomkins}{Kumar
  et~al\mbox{.}}{2010}]{KuNoTo10}
{\sc Kumar, R.}, {\sc Novak, J.}, {\sc and} {\sc Tomkins, A.} 2010.
\newblock Structure and evolution of online social networks.
\newblock In {\em Link Mining: Models, Algorithms, and Applications}. Springer,
  337--357.

\bibitem[\protect\citeauthoryear{Leskovec, Chakrabarti, Kleinberg, and
  Faloutsos}{Leskovec et~al\mbox{.}}{2005}]{LeChKlFa05}
{\sc Leskovec, J.}, {\sc Chakrabarti, D.}, {\sc Kleinberg, J.}, {\sc and} {\sc
  Faloutsos, C.} 2005.
\newblock Realistic, mathematically tractable graph generation and evolution,
  using {Kronecker} multiplication.
\newblock In {\em PKDD 2005}. Springer, 133--145.

\bibitem[\protect\citeauthoryear{Leskovec, Chakrabarti, Kleinberg, Faloutsos,
  and Ghahramani}{Leskovec et~al\mbox{.}}{2010}]{LeChKlFa10}
{\sc Leskovec, J.}, {\sc Chakrabarti, D.}, {\sc Kleinberg, J.}, {\sc Faloutsos,
  C.}, {\sc and} {\sc Ghahramani, Z.} 2010.
\newblock Kronecker graphs: An approach to modeling networks.
\newblock {\em J. Machine Learning Research\/}~{\em 11}, 985--1042.

\bibitem[\protect\citeauthoryear{Leskovec and Faloutsos}{Leskovec and
  Faloutsos}{2007}]{LeFa07}
{\sc Leskovec, J.} {\sc and} {\sc Faloutsos, C.} 2007.
\newblock Scalable modeling of real graphs using kronecker multiplication.
\newblock In {\em ICML '07}. ACM, 497--504.

\bibitem[\protect\citeauthoryear{Mahdian and Xu}{Mahdian and Xu}{2007}]{MaXu07}
{\sc Mahdian, M.} {\sc and} {\sc Xu, Y.} 2007.
\newblock Stochastic kronecker graphs.
\newblock In {\em Algorithms and Models for the Web-Graph}. Springer, 179--186.

\bibitem[\protect\citeauthoryear{Mahdian and Xu}{Mahdian and Xu}{2011}]{MaXu10}
{\sc Mahdian, M.} {\sc and} {\sc Xu, Y.} 2011.
\newblock Stochastic {K}ronecker graphs.
\newblock {\em Random Structures \& Algorithms\/}~{\em 38,\/}~4, 453--466.
\newblock Conference version appeared as \cite{MaXu07}.

\bibitem[\protect\citeauthoryear{McDiarmid}{McDiarmid}{1989}]{Mc89}
{\sc McDiarmid, C.} 1989.
\newblock On the method of bounded differences.
\newblock {\em Surveys in Combinatorics\/}~{\em 141}, 148--188.

\bibitem[\protect\citeauthoryear{Miller, Bliss, and Wolfe}{Miller
  et~al\mbox{.}}{2010}]{MiBlWo10}
{\sc Miller, B.}, {\sc Bliss, N.}, {\sc and} {\sc Wolfe, P.} 2010.
\newblock Subgraph detection using eigenvector {L1} norms.
\newblock In {\em NIPS 2010}. 1633--1641.

\bibitem[\protect\citeauthoryear{Mitzenmacher}{Mitzenmacher}{2003}]{Mi03}
{\sc Mitzenmacher, M.} 2003.
\newblock A brief history of generative models for power law and lognormal
  distributions.
\newblock {\em Internet Mathematics\/}~{\em 1,\/}~2, 226--251.

\bibitem[\protect\citeauthoryear{Mitzenmacher}{Mitzenmacher}{2006}]{Mi06}
{\sc Mitzenmacher, M.} 2006.
\newblock The future of power law research.
\newblock {\em Internet Mathematics\/}~{\em 2,\/}~4, 525--534.

\bibitem[\protect\citeauthoryear{Moreno, Kirshner, Neville, and
  Vishwanathan}{Moreno et~al\mbox{.}}{2010}]{MoKiNeVi10}
{\sc Moreno, S.}, {\sc Kirshner, S.}, {\sc Neville, J.}, {\sc and} {\sc
  Vishwanathan, S. V.~N.} 2010.
\newblock Tied {Kronecker} product graph models to capture variance in network
  populations.
\newblock In {\em Proc. 48th Annual Allerton Conf. on Communication, Control,
  and Computing}. 1137--1144.

\bibitem[\protect\citeauthoryear{Motwani and Raghavan}{Motwani and
  Raghavan}{1995}]{MR}
{\sc Motwani, R.} {\sc and} {\sc Raghavan, P.} 1995.
\newblock {\em Randomized Algorithms}.
\newblock Cambridge University Press.

\bibitem[\protect\citeauthoryear{Pennock, Flake, Lawrence, Glover, and
  Giles}{Pennock et~al\mbox{.}}{2002}]{PeFlLa+02}
{\sc Pennock, D.}, {\sc Flake, G.}, {\sc Lawrence, S.}, {\sc Glover, E.}, {\sc
  and} {\sc Giles, C.~L.} 2002.
\newblock Winners don't take all: Characterizing the competition for links on
  the web.
\newblock {\em PNAS\/}~{\em 99,\/}~8, 5207--5211.

\bibitem[\protect\citeauthoryear{Sala, Cao, Wilson, Zablit, Zheng, and
  Zhao}{Sala et~al\mbox{.}}{2010}]{SaCaWiZa10}
{\sc Sala, A.}, {\sc Cao, L.}, {\sc Wilson, C.}, {\sc Zablit, R.}, {\sc Zheng,
  H.}, {\sc and} {\sc Zhao, B.~Y.} 2010.
\newblock Measurement-calibrated graph models for social network experiments.
\newblock In {\em WWW '10}. ACM, 861--870.

\end{thebibliography}
\end{document}